\documentclass[a4paper,11pt]{article}
\usepackage[english]{babel}
\usepackage{natbib}
\usepackage{anysize}
\marginsize{3cm}{3cm}{3cm}{3cm}
\usepackage{amsfonts,amsmath,amssymb,amsthm,bbm}
\usepackage{graphicx} 
\usepackage{color}
\usepackage{subfigure}

\usepackage[colorlinks=True, citecolor=blue]{hyperref}

\theoremstyle{plain}
\newtheorem{prop}{Proposition}
\theoremstyle{remark}
\newtheorem{rk}{Remark}

\def\beq{\begin{equation}}
\def\eeq{\end{equation}}
\def\fomega{\xi}
\def\cS{\mathcal{S}}
\def\cF{\mathcal{F}}
\newcommand{\IND}{\ensuremath{\mathbbm{1}}}
\def\figdir{figures}
\def\twofig{.48\textwidth}
\def\threefig{.32\textwidth}
\def\bom{{\boldsymbol{\fomega}}}
\def\sym{\text{sym}}

\title{On the zeros of the spectrogram of
  white noise}
\author{
R\'emi Bardenet$^{1}$\footnote{Corresponding author: \href{mailto:remi.bardenet@gmail.com}{remi.bardenet@gmail.com}}, Julien Flamant$^1$, Pierre Chainais$^1$\\
\small $^1$ Univ. Lille, CNRS, Centrale Lille, UMR 9189 -— CRIStAL, 59651 Villeneuve d'Ascq, France \\
}

\date{}

\begin{document}
\maketitle
\begin{abstract}
In a recent paper, \cite{Fla15} has proposed filtering based on the zeros of a
spectrogram, using the short-time Fourier transform and a Gaussian window. His
results are based on empirical observations on the distribution of the zeros of
the spectrogram of white Gaussian noise. These zeros tend to be uniformly spread over the
time-frequency plane, and not to clutter. Our contributions are threefold:
we rigorously define the zeros of the spectrogram of continuous white Gaussian noise, we
explicitly characterize their statistical distribution, and we investigate the
computational and statistical underpinnings of the practical implementation of
signal detection based on the statistics of spectrogram zeros. In
particular, we stress that the zeros of spectrograms of white Gaussian noise
correspond to zeros of Gaussian analytic functions, a topic of recent
independent mathematical interest \citep{HKPV09}.
\end{abstract}

\section{Introduction}
\label{s:intro}
Spectrograms are a cornerstone of time-frequency analysis \citep{Fla98}. They are quadratic time-frequency representations of a
signal \cite[Chapter 4]{Gro01}, associating to each time and frequency a real
number that measures the energy content of a signal at that time and frequency, unlike global-in-time tools such as the Fourier
transform. Since it is natural to expect that there is more energy where there
is more information or signal, most methodologies have focused on detecting and
processing the local maxima of the spectrogram \citep{Coh95, Fla98, Gro01}. Usual techniques include \emph{ridge extraction},
e.g., to identify chirps, or \emph{reassignment} and \emph{synchrosqueezing}, to
better localize the maxima of the spectrogram before further quantitative
analysis. 

In contrast, \cite{Fla15} has recently observed that the locations of the zeros of a
spectrogram in the time-frequency plane almost completely characterize the
spectrogram, and he proposed to use the point pattern formed by the zeros in
filtering and reconstruction of signals in noise. This proposition stems from
the empirical observation that the zeros of the short-time Fourier transform of
white noise are uniformly spread over the time-frequency plane, and tend not to
clutter, as if they repelled each other. In the presence of a signal, zeros are
absent in the time-frequency support of the signal, thus creating large holes
that appear to be very rare when observing pure white noise. This leads to testing
the presence of signal by looking at statistics of the point pattern of zeros,
and trying to identify holes. In this paper, we attempt a formalization of the
approach of \cite{Fla15}. To this purpose, we put together notions of signal
processing, complex analysis, probability, and spatial statistics.

Our contributions are threefold: we rigorously define the zeros of the
spectrogram of continuous white noise, we explicitely characterize their
statistical distribution, and we investigate the computational and statistical
underpinnings of the practical implementation of signal detection. In
particular, we stress that zeros of spectrograms of white noise correspond to
zeros of Gaussian analytic functions, a topic of recent independent mathematical interest \citep{HKPV09}.

In short, our approach starts from the usual definition of white noise as a random
tempered distribution. Using a classical equivalence between the short-time
Fourier transform and the Bargmann transform, we show that the short-time
Fourier transform of white noise can be identified with a random analytic
function, so that we can give a precise meaning to the zeros of the spectrogram of white
noise. It turns out that real and complex Gaussian white noises lead to
recently studied random analytic functions, with completely characterized zeros. We
then investigate how to leverage probabilistic information on these zeros to design
statistical detection procedures. This includes linking probability and
complex analysis results to the discrete implementation of the Fourier
transform.

The rest of the paper is organized as follows. In Section~\ref{s:preliminary},
we introduce the relevant notions of complex analysis, probability, and spatial
statistics. In Section~\ref{s:real}, we characterize the zeros of the short-time
Fourier transform of real white noise, while the complex and the analytical case
are treated in Section~\ref{s:complex}. In Section~\ref{s:stats}, we investigate
the relation between the previous sections and the usual discrete implementation
of the Fourier transform, and we demonstrate a detection task using the spectrogram zeros.

\section{Spectrograms, complex analysis, and point processes}
\label{s:preliminary}
In this section, we survey the relevant notions
from signal processing, probability, and spatial statistics. 

\subsection{The short-time Fourier transform}
\label{s:stft}
Let $f,g\in L^2(\mathbb{R})$, the evaluation at $(u,v)\in\mathbb{R}^2$ of the
short-time Fourier transform (STFT) of $f$ with \emph{window} $g$ reads
\beq
 V_g f (u,v) = \int f(t) \overline{g(t-u)} e^{-2i\pi tv}dt = \langle f, M_v T_u g \rangle,
\label{e:stft}
\eeq
with $\langle \cdot,\cdot\rangle$ denoting the inner product in
$L^2(\mathbb{R})$, $M_vf =
e^{2i\pi v\cdot}f(\cdot)$ and $T_uf = f(\cdot-u)$. We copy our notation from \citep[Chapter 3]{Gro01}, to which
we refer for a thorough introduction. The squared modulus of the STFT
\eqref{e:stft} is called a
\emph{spectrogram}, and it is commonly interpreted as a measure of the
content of the signal $f$ around time $u$ and frequency $v$. In contrast, the
usual Fourier transform only provides the \emph{global} frequency content of
a signal, that is, not localized in time.

The right-hand side of \eqref{e:stft} allows a natural extension of the STFT to
tempered distributions, see \citep[Section 3.1]{Gro01}. This is relevant to us,
as white noise will be defined in Sections~\ref{s:real} and \ref{s:complex} as a random tempered
distribution.

\subsection{The Bargmann transform}
\label{s:bargmann}

Let $a>0$ and consider the Gaussian window $g_a(x) \propto \exp(-\pi a^2 x^2)$, normalized so
that $\Vert g_a\Vert_2=1$. When $a=1$, we drop the subscript and write $g(x) = g_1(x) = 2^{1/4}e^{-\pi x^2}$.

We closely follow the textbook by \cite{Gro01}, only introducing arbitrary window
width, and gather the important result in the following proposition.
\begin{prop}{\cite[Section 3.4]{Gro01}}
Let $f\in L^2(\mathbb{R})$, $u,v\in\mathbb{R}$ and $z=au+i\frac{v}{a}$, then
\begin{eqnarray}
V_{g_a}(f)(u,-v) &\propto&  e^{-i\pi uv}e^{-\frac{\pi}{2}\vert z\vert^2} B\left( f(\cdot/a)
          \right)(z),\label{e:bargmann}
\end{eqnarray}
where the Bargmann transform $B$ is defined by
$$
Bf(z) = 2^{1/4}\int f(t)e^{2\pi tz-\pi t^2-\frac{\pi}{2}z^2}dt.
$$
\end{prop}

\begin{proof}
The particular shape of the window allows us to write
\begin{eqnarray*}
V_{g_a}(f)(u,v) &\propto& \int f(t)e^{-\pi a^2(t-u)^2}e^{-2i\pi tv}dt\\
&=& \int f(t) e^{-\pi a^2t^2}e^{-\pi a^2u^2}e^{2a^2\pi tu}e^{-2i\pi vt}dt\\
&=& e^{-i\pi uv} e^{-\frac{\pi}{2}(a^2u^2+\frac{v^2}{a^2})}\int f(t) e^{-\pi a^2
    t^2}e^{2a\pi t(au-i\frac{v}{a})} e^{-\frac{\pi}{2}(au-i\frac{v}{a})^2}dt.
\end{eqnarray*}
Making the change of variables $s=at$ and denoting
\begin{equation}
z=au+i\frac{v}{a},
\label{e:complexTiling}
\end{equation}
 we obtain
\begin{eqnarray*}
V_{g_a}(f)(u,v) &\propto& e^{-i\pi uv}e^{-\frac{\pi}{2}\vert z\vert^2}\int
                          f\left(\frac{s}{a}\right)e^{-\pi s^2}e^{2\pi s \bar{z}}e^{-\frac{\pi}{2}\bar{z}^2}ds,
\end{eqnarray*}
or equivalently 
\begin{eqnarray}
V_{g_a}(f)(u,-v) &\propto& e^{-i\pi uv}e^{-\frac{\pi}{2}\vert z\vert^2}\int
                          f\left(\frac{s}{a}\right)e^{-\pi s^2}e^{2\pi s
                           z}e^{-\frac{\pi}{2}z^2}ds\nonumber\\
&\propto& e^{-i\pi uv}e^{-\frac{\pi}{2}\vert z\vert^2} B\left( f(\cdot/a)
          \right)(z),\label{e:bargmann}
\end{eqnarray}
where we have defined the Bargmann transform by
$$
Bf(z) = 2^{1/4}\int f(t)e^{2\pi tz-\pi t^2-\frac{\pi}{2}z^2}dt.
$$
\end{proof}
Equation \eqref{e:bargmann} tells us that the zeros of the spectrogram
$u,v\mapsto \vert V_{g_a}(f)(u,v)\vert^2$ are those of the Bargmann transform
of $s\mapsto f(s/a)$. Moreover, Equation~\eqref{e:bargmann} also readily extends
to tempered distributions.

\subsection{Hermite functions}
\label{s:hermite}
Some functions turn out to have a very simple closed-form Bargmann transform. Informally, if we had an
orthonormal basis of $L^2(\mathbb{R})$ formed by such functions, then we could
decompose a white noise onto this basis, and easily compute the STFT of white
noise using closed-form Bargmann transforms. We now introduce Hermite functions,
which will play this exact role in later sections.

Let $H_n$ be the orthonormal polynomials with respect to the Gaussian window $g$, usually called the Hermite
polynomials in the literature \citep{Gau04}. Then, making the change of variables $x'=ax$, it comes
$$
\int H_k(ax)H_\ell(ax)g_a(x)d x \propto \int H_k(x') H_\ell(x')
g(x') d x' =
\delta_{k\ell}.
$$
The Hermite functions $h_{a,k} \propto H_k(a\cdot) \sqrt{g_a(\cdot)}$, normed
so that $\Vert h_{a,k}\Vert_2 =1$, form an orthonormal basis of
$L^2(\mathbb{R})$ \citep{Gau04}. When
$a=1$, we again drop a subscript and denote $h_k=h_{1,k}$. To compute the STFT of
an Hermite function using \eqref{e:bargmann}, first note that for all $s$,
$h_{a,k}(s/a) \propto h_{k}(s)$,
so that
\begin{eqnarray*}
V_{g_a}(h_{a,k})(u,-v) &\propto& e^{-i\pi uv}e^{-\frac{\pi}{2}\vert z\vert^2} B(h_k)(z)\\
&=& e^{-i\pi uv}e^{-\frac{\pi}{2}\vert z\vert^2} \frac{\pi^{k/2}z^k}{\sqrt{k!}},
\end{eqnarray*}
see \cite[Section 3.4]{Gro01} for the last equality. 

\subsection{Point processes on $\mathbb{C}$}
\label{s:spatial}
The zeros of the spectrogram of a random signal form a \emph{point process}. Formally, a point process over $\mathbb{C}$ is a probability distribution over configurations of
points in $\mathbb{C}$, i.e., unordered sets of complex numbers. In particular,
the cardinality of a realization of a point process is random. In this section,
we introduce point processes and basic descriptive statistics.

\subsubsection{Generalities}
The simplest point process over $\mathbb{C}$ is the Poisson point process with constant rate
$\lambda\in\mathbb{R_+}$. It is defined as the unique point process
such that, for any $B\in\mathbb{C}$ with finite Lebesgue measure $\vert B\vert$,
$(i)$ the number of points in
$B$ is a Poisson random variable with mean $\lambda\vert B\vert$, and $(ii)$
conditionally on the number of points in $B$, the points are drawn independently
from the uniform measure on $B$. For existence and further properties, see e.g.
\cite[Chapter 3]{MoWa03}. 

More general point processes can be characterized by their $k$-point correlation functions
$\rho^{(k)}$ for $k\geq 1$, informally defined by
\begin{equation}
\rho^{(k)}(x_1,\dots,x_k)dx_1\dots dx_k = \mathbb{P}\begin{pmatrix}\text{There are at least $k$ points, one in each of the}\\
\text{infinitesimal balls $B(x_i, dx_i), i=1,\dots,k$}\end{pmatrix},
\label{e:correlationFunctions}
\end{equation}
for all $x_1,\dots,x_k$ in $\mathbb{C}$, see \cite[Section 5.4]{DaVe03} for a
rigorous treatment. Of particular interest to us will be the first and
second-order interaction between the points in a realization of a point process,
encoded by $\rho^{(1)}$ and $\rho^{(2)}$, respectively.

The first order correlation function $\rho^{(1)}$ is often called the intensity of the point process, for it yields, when
integrated over a Borel set $B\subset \mathbb{C}$, the average number of points
falling in $B$ under the point process distribution. For the Poisson point
process with constant rate $\lambda$, for instance, the intensity is precisely
$\lambda$, and thus constant over $\mathbb{C}$. 

The two-point correlation function $\rho^{(2)}$ is often renormalized to obtain
the so-called \emph{pair correlation function} $$ g(x,y) =
\frac{\rho^{(2)}(x,y)}{\rho^{(1)}(x)\rho^{(1)}(y)},$$ see \cite[Chapter
4]{MoWa03}. For a Poisson point process with constant rate, $g$ is
identically $1$. When $g(x,y)>1$, \eqref{e:correlationFunctions} indicates that
pairs are more likely to occur around $(x,y)$ than under a Poisson process with the same intensity
function. Similarly $g(x,y)<1$ indicates that pairs are less likely
to occur. Finally, when the point process is both stationary (i.e., invariant
to translations) and isotropic (i.e., invariant to rotations), then $g$ only
depends on the distance $r=\vert x-y\vert$, and we denote it by $g_0(r)=g(x,y)$.

\subsubsection{The Ginibre ensemble}
\label{e:ginibre}
We give here another example of a point process on $\mathbb{C}$, in
order to demonstrate a non-constant pair correlation function. If
there exists a function $\kappa:\mathbb{C}\times\mathbb{C}\rightarrow \mathbb{C}$
such that the correlation functions \eqref{e:correlationFunctions} with 
\begin{equation}
\rho^{(k)}(x_1,\dots,x_k) = \det \big[\kappa(x_i,x_j)\big]_{1\leq i,j\leq k}
\label{e:dppCorrelationFunctions}
\end{equation}
consistently define a point process, then this point process is called a
\emph{determinantal} point process (DPP) with kernel $\kappa$. DPPs were first
introduced by \cite{Mac75}, and we refer the reader to \citep{HKPV06,LaMoRu14}
for modern introductions and conditions of
existence. A classical example of DPP
over $\mathbb{C}$ is the infinite Ginibre ensemble. It is defined by its kernel
$$\kappa^{\text{Gin}}(z,w) = e^{-\frac{\pi}{2}\vert z\vert^2} e^{\pi z\bar w}
e^{-\frac{\pi}{2}\vert w\vert^2}.$$
The Ginibre ensemble is stationary
and isotropic, its intensity is constant equal to 1, and its pair correlation is
$$
g_0^{\text{Gin}}(r) = 1 - e^{-\pi r^2},
$$
see \cite[Section 4.3.7]{HKPV09} for these properties, noting that our version
is rescaled to have unit intensity. We also plot
$g_0^{\text{Gin}}$ in Figure~\ref{f:rho}. Importantly for us,
$g_0^{\text{Gin}}(r)<1$ for all $r>0$, which shows that Ginibre is a
\emph{repulsive} point process: pairs are less likely than Poisson at all
scales, which we can interpret as points in a realization repelling each other.
Finally, we note that by definition \eqref{e:dppCorrelationFunctions}, if a DPP
is stationary and isotropic, and if it has an Hermitian kernel, that is $\kappa(x,y)=\overline{\kappa(y,x)}$, then $g_0<1$. 

\subsubsection{Functional statistics}
\label{s:LAndK}
We will need to investigate how repulsive a stationary and isotropic point
process on $\mathbb{C}$ like Ginibre is, given one of its realizations over a compact
window of observation. While estimators of $g_0$ have been investigated
\cite[Section 4.3]{MoWa03}, practitioners usually prefer estimating Ripley's $K$
function
$$
K(r) = 2\pi\int_{0}^r t g_0(t)dt, \quad r>0,
$$
and then the so-called \emph{variance-stabilized} $L$ functional statistic
$$
L(r) = \sqrt{K(r)/\pi},
$$
which equals $r$ for a unit rate Poisson process. $K$ is proportional to the
expected number of pairs at distance smaller than $r$. Estimating $K$ from data
is thus relatively straightforward and involves counting pairs distant from a
collection of values of $r$. Furthermore, sophisticated edge corrections have
been proposed to take into account the fact that the observation window is
necessarily bounded \cite[Section 4.3]{MoWa03}. Estimating $L$ after one has
obtained an estimate of $K$ is then
straightforward. Plotting the estimated $K$ or $L$ as a
function of $r$ allows identification of scales at which the point process is
repulsive, in the sense that we can observe a lack of pairs within a given
distance compared to a Poisson process. For instance, we plot in
Figure~\ref{f:L} the function $r\mapsto L(r)-r$ for Ginibre: there is a clear
lack of pairs at small scales, compared to the constant zero of a Poisson
process. 

\cite[Section 4.2]{MoWa03} cover many more functional statistics for stationary point processes. In particular, we mention for future reference the
so-called \emph{empty space function} $F$ and the \emph{nearest neighbour}
function $G$. For $r>0$, $F(r)$ is
defined as the probability that a ball centered at $0$ and with radius $r$
contains at least one point. Stationarity implies that the center of the ball
can be chosen arbitrarily, and $F$ thus encodes the distribution of hole sizes
in the point process. Similarly, $G$ is the cumulative distribution function of
the distance from a typical random point of the point process to its nearest
neighbour in the point process.

\section{The spectrogram of real white noise}
\label{s:real}
In this section, we define real white noise, and examine the zeros of its
spectrogram. 
\subsection{Definitions}

To define white noise, we closely follow \cite[Chapter 2.1]{HOUZ10} through a
classical approach that does not require defining Brownian motion first. We denote by $\cS=\cS(\mathbb{R})$
the Schwartz space of rapidly decaying smooth complex-valued functions of a real
variable. The dual $\cS'=\cS'(\mathbb{R})$, equipped with the weak-star
topology, is the space of \emph{tempered distributions}. The topology yields the
Borel sigma-algebra $\mathcal{B}(\cS')$ on $\cS'$. Now, the Bochner-Minlos
theorem \citep[Theorem 2.1.1]{HOUZ10} states that there exists a unique probability
measure $\mu_1$ on $(\cS',\mathcal{B}(\cS'))$ such that 
\beq
\forall \phi\in \cS, \quad \mathbb{E}_{\mu_1}e^{i\langle \cdot,\phi \rangle} =
e^{-\frac{1}{2}\Vert \phi\Vert_2^2}.
\label{e:bochner}
\eeq
We call this measure white noise, and $(\cS',B(\cS'),\mu_1)$ the white noise
probability space. In particular, \eqref{e:bochner}
implies that for a random variable\footnote{We use the term \emph{random variable}, but it is also customary to call $\fomega$ a \emph{generalized random process} in the literature.} with distribution
$\mu_1$ and a set of real-valued orthonormal functions $\varphi_1,\dots,\varphi_p$ in $\cS$, the vector
$(\langle \fomega_1,\varphi_1 \rangle,\dots,\langle \fomega_p,\varphi_p \rangle)$
follows a real multivariate Gaussian, with mean zero and identity covariance matrix,
see \cite[Lemma 2.1.2]{HOUZ10}. This is in accordance with the usual heuristic of white noise having a Dirac delta covariance function.

Let $\fomega$ be a random variable with distribution $\mu_1$.
If $g\in\cS$, then $(u,v)\mapsto M_v T_u g$ is in $\cS$, so that we can define
the STFT of $\fomega$ as the random function
$$ u,v\mapsto \langle \fomega,  M_v T_u g \rangle.$$
From now on, we restrict ourselves to the Gaussian window $g(x) =
2^{1/4}e^{-\pi x^2}$, normalized so that $\Vert g\Vert_2 = 1$. We are interested
in defining and studying the zeros of the spectrogram
\beq
\label{e:spectrogram}
S:u,v\mapsto \vert \langle \fomega,  M_v T_u g \rangle\vert^2.
\eeq

\subsection{Characterizing the zeros}
\label{e:zerosSymmetricGAF}
We work in two steps: in Proposition~\ref{p:series}, we identify each value
$S(u,v)$ in \eqref{e:spectrogram} as a limit in $L^2(\mu_1)$, and we then
show in Proposition~\ref{p:symmetricPlanarGAF} that the resulting random field
defines an entire function, the zeros of which are known.
\begin{prop}
Let $u,v\in\mathbb{R}^2$, and write $z=u+iv\in\mathbb{C}$. Then
\beq
\langle \fomega,  M_v T_u g \rangle = \sqrt{\pi} e^{i\pi uv}e^{-\frac{\pi}{2}\vert z\vert^2} \sum_{k=0}^{\infty}\langle \fomega,h_k
\rangle \frac{\pi^{k/2}z^k}{\sqrt{k!}} 
\label{e:series}
\eeq
where $(h_k)$ denote the orthonormal Hermite functions \cite[Section 2.2.1]{HOUZ10}, and convergence is in $L^2(\mu_1)$.
\label{p:series}
\end{prop}
\begin{rk}
\label{r:zeros}
Note that in Proposition~\ref{p:series}, $u$ and $v$ are fixed, and the equality
is a limit in $L^2(\mu_1)$. It is still too early to identify the zeros of
the left-hand side to the zeros of the right-hand side.
\end{rk}
\begin{rk}
Note that our choice of the window $g(x) = 2^{1/4}e^{-\pi x^2}$ is made to
simplify expressions. The proof of Proposition~\ref{p:series}, along with
Sections~\ref{s:hermite} and \ref{s:bargmann}, immediately yield that for a non-unit Gaussian window $g_a(x) \propto \exp(-\pi
a^2 x^2)$, Proposition~\ref{p:series} is unchanged, provided that $z$ is defined
as $z = au + iv/a$ and a constant is prepended to the RHS of
\eqref{e:series}. In other words, given a particular value of $a$, it is always possible to
dilate/squeeze the time-frequency axes to obtain the results detailed
here for $a = 1$. 
\end{rk}
\begin{proof}
Let $u,v\in\mathbb{R}^2$. Decomposing $M_v T_u g$ in the Hermite basis $(h_k)$ of
$L^2(\mathbb{R})$, it comes
\begin{eqnarray}
\langle \fomega, M_v T_u g \rangle &=& \sum_{k=0}^\infty \langle \fomega,h_k
  \rangle \langle  M_v T_u g,h_k\rangle\nonumber\\
&=&  \sum_{k=0}^\infty \langle \fomega,h_k
  \rangle \overline{V_g(h_k)(u,v)}
\label{e:decomp}
\end{eqnarray}
where the limits are in $L^2(\mu_1)$. The STFT of Hermite functions is
well-known, see e.g. the proof of \cite[Proposition 3.4.4]{Gro01} or our Section~\ref{s:bargmann}, and it reads
\beq 
V_g(h_k)(u,v) = e^{-i\pi
  uv}e^{-\frac{\pi}{2}(u^2+v^2)}\frac{\pi^{k/2}}{\sqrt{k!}}(u-iv)^k.
\label{e:HermiteSTFT}
\eeq
Plugging \eqref{e:HermiteSTFT} into \eqref{e:decomp} yields the result.
\end{proof}

Now we focus on the regularity of the right-hand side of \eqref{e:series}.
\begin{prop}
\label{p:symmetricPlanarGAF}
The random series
\beq
\label{e:symmetricPlanarGAF}
\sum_{k=0}^{\infty}\langle \fomega,h_k
\rangle \frac{\pi^{k/2}z^k}{\sqrt{k!}}
\eeq
$\mu_1$-almost surely defines an entire function. 
\end{prop}
\begin{proof}
By \cite[Lemma 2.1.2]{HOUZ10}, ($\langle \fomega,h_k
\rangle)_{k\geq 0}$ are i.i.d. unit real Gaussians. We then apply the first part
of \cite[Lemma 2.2.3]{HKPV09}.
\end{proof}
Since both $L^2$ and almost sure convergence imply convergence in probability,
$L^2$ and almost sure limits have to be the same. In particular,
Propositions~\ref{p:series} and \ref{p:symmetricPlanarGAF} together yield that the
distribution of the zeros of the spectrogram $S$ in \eqref{e:spectrogram} is
the same as the distribution of the zeros of the random entire function
\eqref{e:symmetricPlanarGAF}. This answers Remark~\ref{r:zeros}. In particular,
we now know that the zeros of $S$ are isolated.

The entire function in \eqref{e:symmetricPlanarGAF} is called the
\emph{symmetric planar Gaussian analytic function} (GAF), and a few of its properties are known
\citep{Fel13}. However, its zeros do not define a stationary point process. In
particular, a portion of the zeros concentrate on the real axis, see
Figure~\ref{f:symmetricPlanarGAF}. Intuitively, one can approximate the zeros of
\eqref{e:symmetricPlanarGAF} by the zeros of the random polynomial obtained from
truncating the series. The resulting polynomial has real coefficients, and it is
thus expected to have real zeros as well as pairs of conjugate complex zeros. As a side note, the number of real zeros is a
topic of study on its own, see e.g. \citep{ScMa08}. 

\begin{figure}
\subfigure[Real white noise/symmetric GAF]{
\includegraphics[width=\twofig]{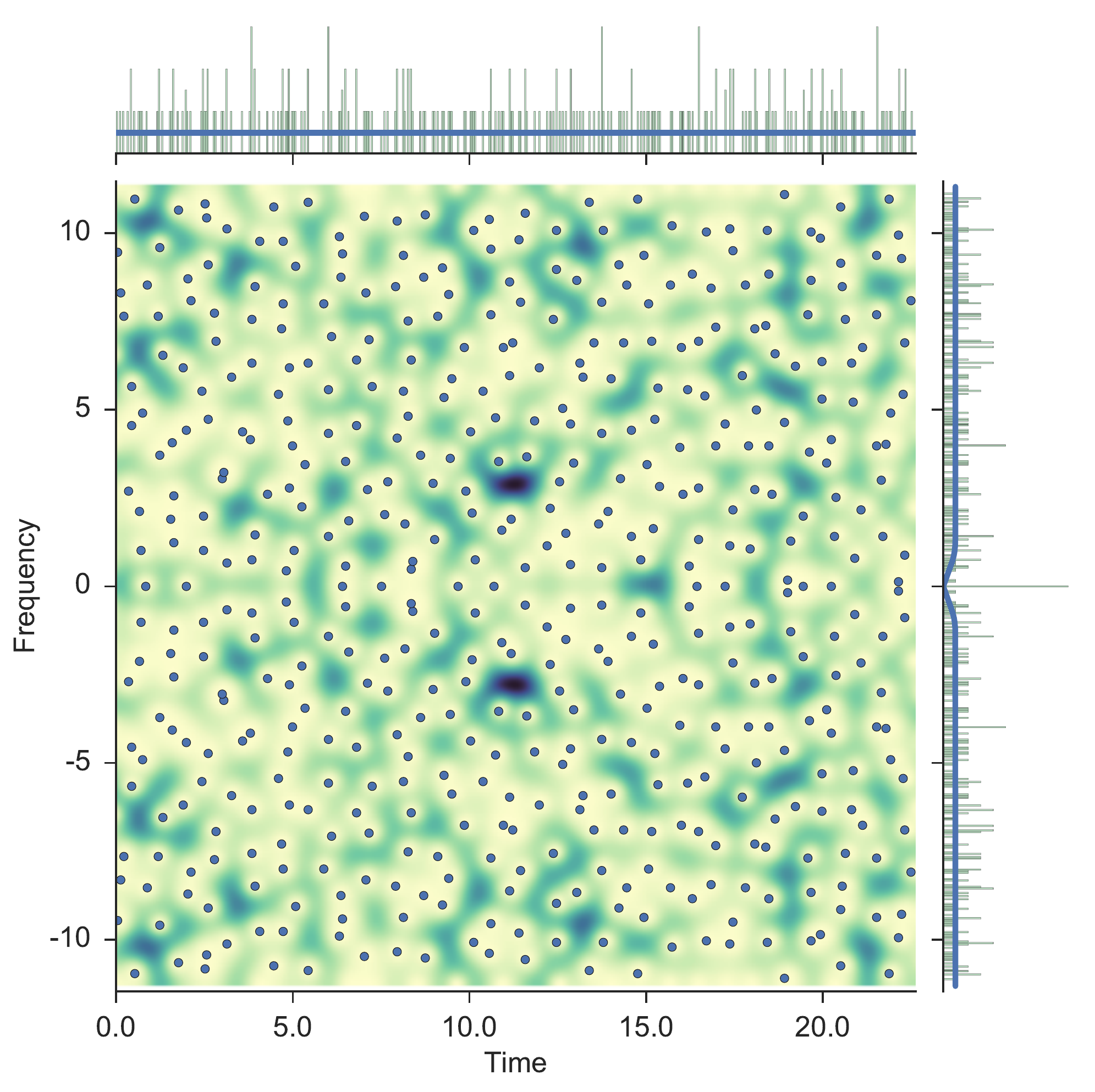}
\label{f:symmetricPlanarGAF}
}
\subfigure[Complex white noise/planar GAF]{
\includegraphics[width=\twofig]{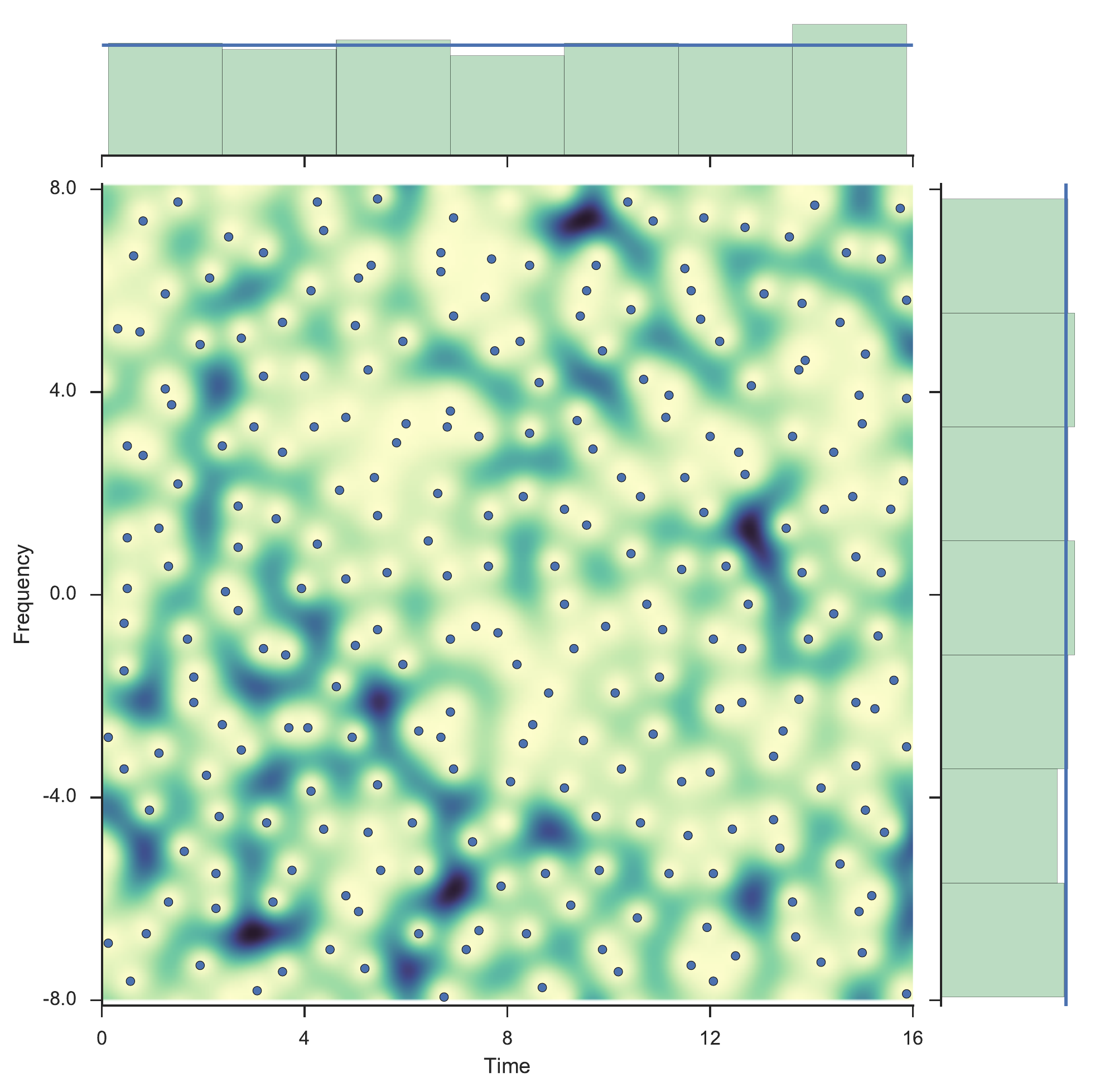}
\label{f:complexPlanarGAF}
}
\caption{The spectrogram of (a) a realization of real white noise, and (b) a
  realization of complex white noise. The right and top plots on each panel show
marginal histograms, superimposed with the theoretical marginal density, see
text for details.}
\label{f:GAFs}
\end{figure}

Coming back to our problem of detecting signals, this non-stationarity
makes it uneasy to approach via traditional spatial statistics techniques, which
often assume some degree of stationarity. However, there is a stationary point
process that is a good approximation for the zeros of the symmetric planar GAF,
and that has been studied in depth. This point process is the zeros of the
\emph{planar GAF}, the entire function corresponding to the STFT of complex
white noise.

\section{The case of complex white noise}
\label{s:complex}
We now introduce the planar GAF, and explain why its zeros are a
good approximation to those of the symmetric planar GAF. In other words, we
justify why the spectrogram of the real white Gaussian noise can be approximated
by that of the complex white Gaussian noise. We conclude by
considering the analytic white noise. 

\subsection{Definitions}
\label{s:complexWhiteNoise}
Consider the \emph{two-dimensional white noise} of \cite[Section
2.1.2]{HOUZ10}, that is, the space $\cS'\times
\cS'$, with the Borel $\sigma$-algebra associated to the product weak star topology, and measure
$\mu_1\times\mu_1$. A draw $\bom=(\fomega_1,\fomega_2)\sim\mu_1\times\mu_1$
consists of two independent white noises. Letting
$\boldsymbol{\phi}=(\phi_1,\phi_2)$ in
$\cS\times \cS$, we \emph{define} the smoothed complex
white noise as in \cite[Exercise 2.26]{HOUZ10} through 
$$ 
w(\boldsymbol{\phi},\fomega) =  \langle\fomega_1, \phi_1\rangle + i \langle\fomega_2, \phi_2\rangle, 
$$
where $\bom\sim\mu_1\times\mu_1 $. It is called ``smoothed'' because we define
it using a pair of test functions $\boldsymbol{\phi}$, which will be enough for
our purpose. Note also that in signal processing, this is typically called a
\emph{proper} or \emph{circular} Gaussian white noise \citep{PiBo97}.

Now, if we let both test functions be $t\mapsto M_v T_u g$, we recover what can
reasonably be called the STFT of complex white noise 
\beq
u,v\mapsto \langle\fomega_1, M_v T_ug\rangle + i \langle\fomega_2, M_v T_u
g\rangle.
\label{e:complexSTFT}
\eeq
\subsection{Characterizing the zeros}
\label{e:zerosPlanarGAF}
The same arguments as in the proofs of Propositions~\ref{p:series} and
\ref{p:symmetricPlanarGAF} lead to
\begin{prop}
\label{p:planarGAF}
With $\mu_1\times\mu_1$ probability $1$, the zeros of the STFT
\eqref{e:complexSTFT} are those of the entire function
\beq
\label{e:planarGAF}
\frac{1}{\sqrt{2}} \sum_{k=0}^\infty\left( \langle \fomega_1,h_k \rangle + i\langle \fomega_2, h_k
  \rangle \right) \frac{\pi^{k/2}z^{k}}{\sqrt{k!}},
\eeq
where $z=u+iv$.
\end{prop}
We note that under $\mu_1\times\mu_1$, the random variables $2^{-1/2}(\langle \fomega_1,h_k \rangle + i\langle \fomega_2, h_k
  \rangle)$ are i.i.d. unit complex Gaussians, and the entire function
  \eqref{e:planarGAF} is called the planar Gaussian analytic function in the
  literature. In particular, the planar
  GAF is one of the three fundamental GAFs in the monograph of \cite{HKPV09}, and more is known about its zeros than for the symmetric planar GAF in
  Proposition~\ref{p:symmetricPlanarGAF}. We group some known results in
  Proposition~\ref{p:propertiesOfPlanarGAF}, selecting results that could be of statistical use in signal processing.
\begin{prop}[\cite{HKPV09,Nis10}]
The planar GAF satisfies the following properties:
\begin{enumerate} 
\item The distribution of its zeros is invariant to rotations
  and translations in the complex plane \cite[Proposition 2.3.7]{HKPV09}. In
  particular, it is a stationary point process.
\item Its correlation functions are known \cite[Corollary 3.4.2]{HKPV09}. In
  particular, the intensity in constant equal to $1$, and with the notation of Section~\ref{s:spatial}, for $z,z'\in\mathbb{C}$ such that $\vert z-z'\vert=r$, the pair correlation function reads
\beq
\label{e:pairCorrelation}
\rho^{(2)}(z,z') = g_0(r) = 
\frac{\left[\sinh^2\left(\frac{\pi r^2}{2}\right) + \frac{\pi^2r^4}{4}\right]\cosh\left(\frac{\pi r^2}{2}\right) - \pi r^2\sinh(\frac{\pi r^2}{2})}{\sinh^3\left(\frac{\pi r^2}{2}\right)}.
\eeq
\item The hole probability 
$$
p_r = \mathbb{P}(\text{no points in the disk
    centered at $0$ and with radius $r$})
$$
scales as 
\beq
r^{-4}\log p_r \rightarrow -3e^{2}/4
\label{p:holeProba}
\eeq
as $r\rightarrow+\infty$ \citep{Nis10}.
\end{enumerate}
\label{p:propertiesOfPlanarGAF}
\end{prop}

\begin{figure}
\subfigure[Pair correlation functions $g_0$]{
\includegraphics[width=\twofig]{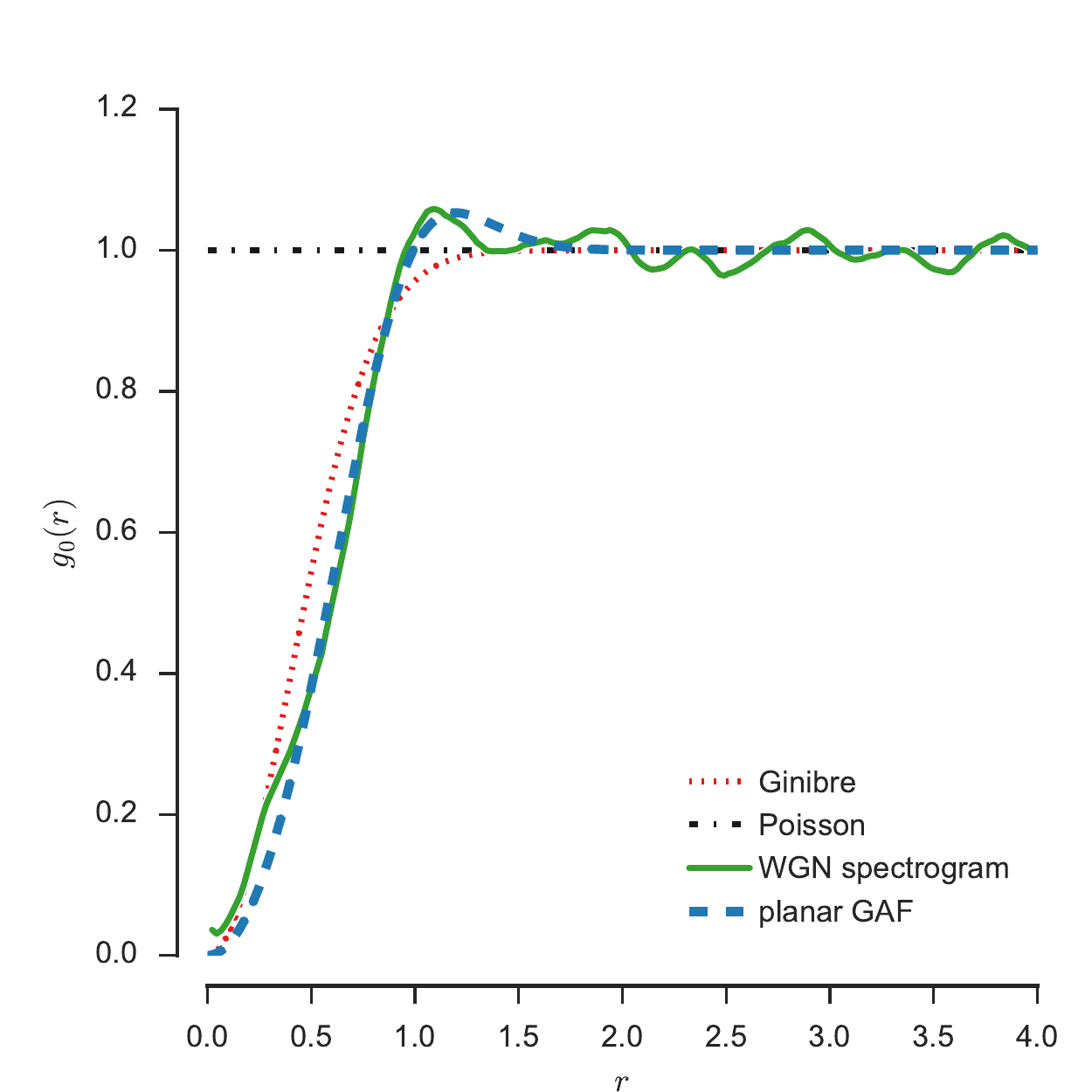}
\label{f:rho}
}
\subfigure[$r\mapsto L(r)-r$ functional statistics]{
\includegraphics[width=\twofig]{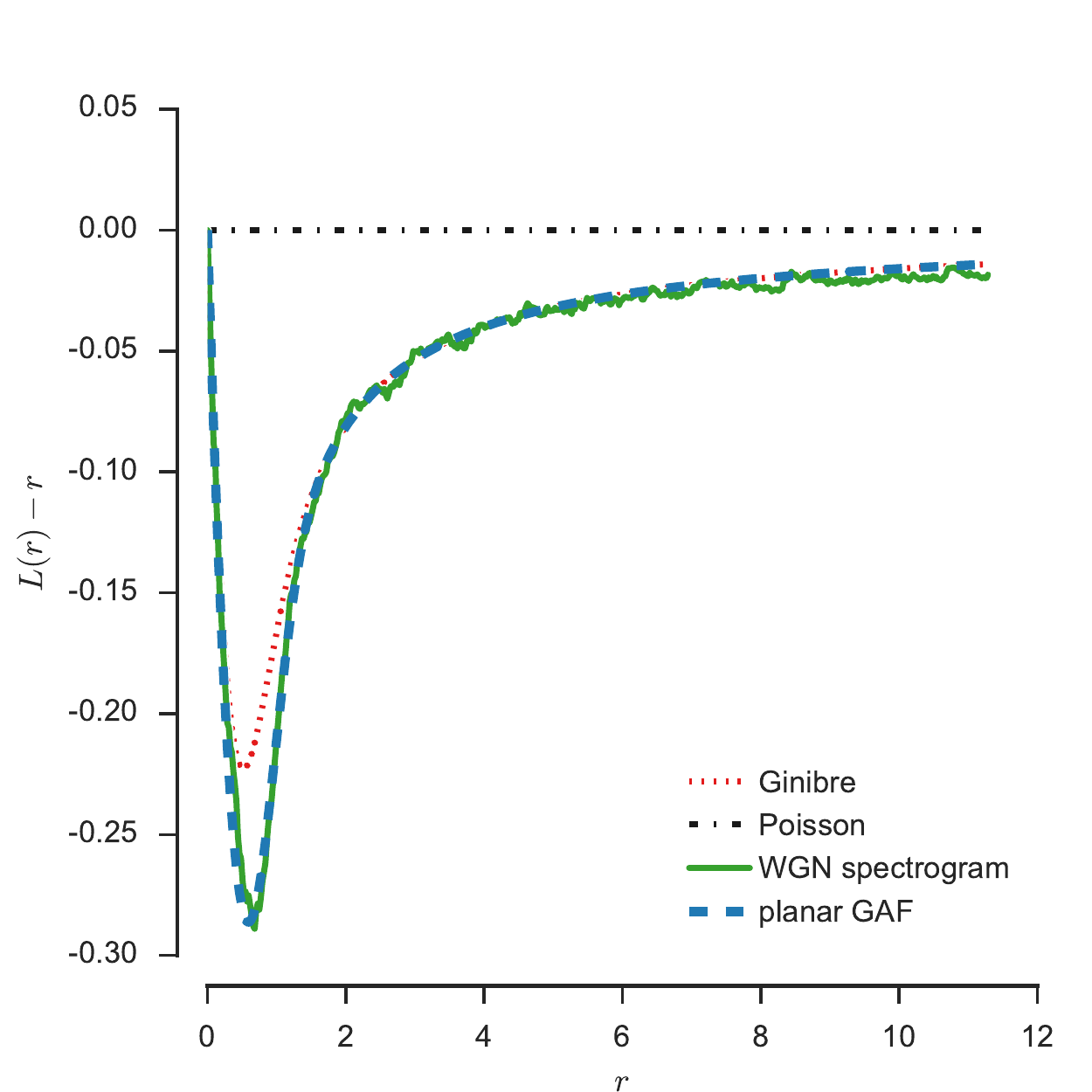}
\label{f:L}
}
\caption{Comparison of the Ginibre point process, the zeros of the planar GAF,
  and a realization of the zeros of the spectrogram of complex white noise, using
  (a) pair correlation functions $g_0$, and (b) the $L$ functional statistic, see Section~\ref{s:spatial} for definitions.}
\label{f:functionalStatistics}
\end{figure}

Figure~\ref{f:functionalStatistics} illustrates
Proposition~\ref{p:propertiesOfPlanarGAF}. We plot the pair correlation function
\eqref{e:pairCorrelation} of the planar GAF, along with the pair correlation
functions of the Poisson and Ginibre point processes introduced in
Section~\ref{s:spatial}. We also superimpose an estimate of $g_0$
obtained from the spectrogram of a realization of a complex white noise, see
Section~\ref{s:stats} for computational procedures. Finally, we also plot the
$L$ functional statistic for the same point processes, as introduced in Section~\ref{s:spatial}.

Both the planar GAF and Ginibre are repulsive at small scales, but the planar
GAF alone has a small ring of attractivity around $r=1$, well visible in
Figure~\ref{f:rho}. This implies that the zeros of the planar GAF cannot be a
DPP with Hermitian kernel, as introduced in
Section~\ref{e:ginibre}, unlike what we and \cite{Fla17} may have intuited. DPPs
were indeed a good candidate for the zeros, as they are repulsive point
processes and naturally relate to reproducing kernel Hilbert spaces, such as
those behind the STFT \cite[Theorem 3.4.2]{Gro01}. But the zeros of the planar
GAF show no repulsion at large scales, and more importantly the pair correlation
function \eqref{e:pairCorrelation} is larger than $1$ around $r=1$, while the
pair correlation of a DPP with Hermitian kernel cannot exceed 1 by definition
\eqref{e:dppCorrelationFunctions}. Note that strictly speaking, it is still
possible that the zeros of the planar GAF are a DPP with a non-Hermitian kernel.

Even if they are not a DPP with Hermitian kernel, the zeros of the planar GAF
are often compared to the Ginibre ensemble, which is a DPP and is also invariant
to isometries of the plane \cite[Section 4.3.7]{HKPV09}. In particular, the decay of the log hole probability \eqref{p:holeProba} is also
in $r^{4}$ for the Ginibre process \cite[Proposition 7.2.1]{HKPV09}. This is to
be compared to the slower decay in $r^2$ of a Poisson process with constant
rate. This is an indication that locally, the zeros of the planar GAF and the
Ginibre ensemble are similarly rigid or regularly spread, and that both are more
rigid than Poisson. There are other intriguing similarities between the two
point processes, see \citep{KrVi14}, where Ginibre is shown to be the zeros of a
GAF with a randomized kernel.

\subsection{The zeros of the planar GAF approximate those of the symmetric
  planar GAF}
  \label{s:GAFapproxSYM}
  
To sum up, the spectrogram of real white noise is described by the symmetric
planar GAF, but the zeros of the planar GAF are more amenable to further
statistical processing. In this section, we survey results by \cite{Fel13} and
\cite{Pro96} that support approximating the zeros of the symmetric planar GAF by
those of the planar GAF.

The zeros of the symmetric planar GAF \eqref{e:symmetricPlanarGAF} have the same
distribution as the zeros of
\beq
f_{\sym}(z) = e^{-\frac{\pi}{2}z^2}\sum_{k=0}^\infty
\frac{a_k}{\sqrt{k!}}\pi^{k/2}z^k,
\label{e:scaledSymmetricGAF}
\eeq
where $a_k$ are i.i.d. unit real Gaussians. Note that the covariance kernel of $f_\sym$ is
\begin{eqnarray*}
K_{\sym}(z,w) &\triangleq& \mathbb{E} f_\sym(z)\overline{f_\sym(w)}\\
&=& e^{-\frac{\pi}{2}z^2} e^{-\frac{\pi}{2}\bar{w}^2} e^{\pi z\bar{w}}\\
&=& e^{-\frac{\pi}{2}(z-\bar{w})^2}.
\end{eqnarray*}
This hints some invariance of $f_\sym$ to translations along the real axis. By a limiting argument, see e.g. \cite[Lemma
  2.3.3]{HKPV09}, \eqref{e:scaledSymmetricGAF} is indeed a stationary symmetric
GAF in the sense of \cite{Fel13}. Namely, for any $n$, any
$z_1,\dots,z_n$, and any $t\in\mathbb{R}$,
$\left(f_\sym(z_1+t),\dots,f_\sym(z_n+t)\right)$ has the same distribution as
$\left( f_\sym(z_1),\dots,f_\sym(z_n) \right)$.

\cite{Fel13} derives the intensity of the zeros of general stationary
symmetric GAFs. More precisely, let $n_{\sym}(B)$ be the random number of zeros
of $f_\sym$ in a Borel set $B\subset\mathbb{C}$, she says that there exists a
so-called \emph{horizontal counting measure} $\nu_\sym$ s.t., almost surely, we have the weak convergence of measures
$$
\nu_\sym(A) = \lim_{T\rightarrow\infty}\frac{n_\sym([0,T]\times A)}{T},
$$
where $A$ is a Borel set on the vertical axis. In other
words, $\nu_\sym$ characterizes the density of zeros averaged across the
horizontal axis. For our symmetric planar GAF \eqref{e:scaledSymmetricGAF},
\cite[Theorem 1]{Fel13} yields
% \beq
% \nu_\sym(y)dy = dS(y) +
% \frac{1}{4\pi}\sqrt{\frac{\psi''(0)}{\psi(0)}}\delta_0,
% \eeq
% where 
% $$
% \psi(y)\triangleq K_{\sym}(x+iy, x+iy) = e^{-\frac{\pi}{2}(2iy)^2} = e^{2\pi
%   y^2} 
% $$
\begin{eqnarray}
\nu_\sym(A) &=& \int_{A} \left[ dS(y) + \delta_0\right],
\label{e:countingMeasureSym}
\end{eqnarray}
where
$$
S(y)  = \frac{y}{\sqrt{1-e^{-4\pi y^2}}}.
$$
%= \left[\frac{1}{4\pi}\frac{\psi'(y)}{\sqrt{\psi(y)^2-\psi(0)^2}}\right]
Equation \eqref{e:countingMeasureSym} is the sum of a continuous component and a
Dirac mass at $0$. The Dirac mass relates to the accumulation of zeros on the
real axis discussed in Section~\ref{s:real}. The numerator of the continuous
part $S$ is the unnormalized cumulative density of a uniform distribution, and the
denominator quickly converges to $1$ as $y$ grows. 

Now compare \eqref{e:countingMeasureSym} to the horizontal
counting measure of the zeros of the planar GAF, which is simply the uniform
$dy$, without any atom, see e.g. \cite[Theorem 1]{Fel13} again. We observe that the two counting
measures are quickly approximately equal, as one goes away from the real axis.
More precisely, for $A\subset [1,+\infty)$, the ratio of $S(A)$ by the Lebesgue
measure of $A$ is within $2\cdot 10^{-6}$ of $1$. For Gaussian windows of
arbitrary width, the change of variables \eqref{e:complexTiling} yields that
the approximation is tight for $\text{Im}(z)\geq a$. This is no obstacle in
signal processing practice, as spectrograms are never considered close to the
real axis, where 'close' is defined by the spread of the observation window in
frequency, which is of order $a$, see Section~\ref{s:analytic}. We also plot the
densities of the continuous part of both measures in Figure~\ref{f:GAFs}. The
Dirac mass of the symmetric planar GAF corresponds to the subset of zeros on the
real axis. 

A natural question is whether the approximation is also accurate for
higher-order interactions in the two point processes. This question can be
addressed by comparing $k$-point correlation functions. The case of the planar
GAF was derived by \cite{Han98}, and closed-form formulas are derived for the
symmetric planar GAF in \cite[Equation (12)]{Pro96}. The latter are not easy to
interpret as they involve nonstandard combinatorial combinations of matrix
coefficients. Still, \cite[Equation 25]{Pro96} shows that when $\text{Im}(z)\gg 0$, the
$k$-point correlation functions of the zeros of the symmetric planar GAF are
well approximated by those of the zeros of the planar GAF. 

To conclude, the distribution of the zeros of the STFT of real white Gaussian
noise is well approximated by that of complex white Gaussian noise, as long as
the observation window is sufficiently far from the time axis. 
% Rather than attempt a technical estimation of the difference between
% the $k$-point correlations of the two sets of zeros, we replicate in
% Figure~\ref{f:} a figure in

% \citep{Pro96} that shows the pair correlation function of the symmetric
% planar GAF on a horizontal strip $\mathbb{R}\times {iy}$ for several values of
% $y$. \textcolor{red}{Comment on how fast the approximation is accurate.}

\subsection{On the analytic white noise}
\label{s:analytic}
A real-valued function $f\in L^2$ has an Hermitian Fourier transform. In signal
processing, it is thus common to cancel out the negative frequencies of a
real-valued signal $f \in L^2$ by defining a complex-valued associated function called its \emph{analytic signal},
\begin{equation}
f^+(x) = 2 \cF^{-1}(\IND_{\mathbb{R}_+}\cF f)(x), \forall x\in\mathbb{R}.
\end{equation}
where $\cF$ is the usual Fourier transform. The term ``analytic'' is related to
the alternative definition of $f^+$ as the boundary function of a particular holomorphic function on the lower half of the
complex plane, see e.g. \cite[Section 2.1]{Pug82} for a concise and rigorous treatment. In signal processing practice, beyond removing redundant
frequencies, the modulus and argument of $f^+$ have meaningful interpretations
for elementary signals
 \citep{Pic97}. Since our initial goal is to understand the behaviour of the zeros of a
real white noise, it is thus tempting to define and consider an
analytic white noise to represent this real white noise. If this approach led to
a simple statistical characterization of zeros, then we would avoid the
approximation by the complex white noise of Section~\ref{s:complexWhiteNoise}.

While folklore has it that the analytic white noise is the circular white noise of Section~\ref{s:complexWhiteNoise}, this is not the case for the most natural
  definition of the analytic signal of a distribution. Following
  \citep[Section 3.3]{Pug82}, we define in this paper the analytic white noise by its action on $L^2$: letting $\fomega\sim\mu_1$ be a real white noise\footnote{As a side
    note, \cite[Section 3]{Pug82} investigates the random field that would be the formal equivalent to the holomorphic continuation of the classical analytic
signal of a function in $L^2$. But this time, the limit on the real axis is
rather ill-behaved.}, we take
\begin{equation}
\langle \fomega^+,f \rangle \triangleq 2 \langle \fomega,
\cF^{-1}(\IND_{\mathbb{R}_+}\cF f)\rangle, \qquad\forall f\in L^2.
\label{e:analyticWhiteNoise}
\end{equation}
For our purpose, it is enough to consider $\fomega^+$
through its action \eqref{e:analyticWhiteNoise}. In particular, if we want to
follow the lines of Sections~\ref{s:real} and \ref{s:complex} and identify the general term of
a random series corresponding to the STFT of
$\fomega^+$, we need an orthonormal basis $(\zeta_k)$ of $L^2$ and a window $g$ such that 
\begin{equation}
\langle \zeta_k, \cF^{-1}(\IND_{\mathbb{R}_+}\cF M_vT_u g)\rangle 
\label{e:startingPoint}
\end{equation}
is known in closed-form and simple enough. Hermite functions and the Gaussian
window definitely do not satisfy our criteria anymore, and we leave this existence as an
open question. Still, we have the following heuristic argument: when $g$ is the
unit-norm Gaussian, \eqref{e:startingPoint} becomes
\begin{equation}
\langle \zeta_k, \cF^{-1}(\IND_{\mathbb{R}_+}T_vM_{-u} g)\rangle,
\label{e:desiderata}
\end{equation}
so that when $v$ is
large enough, say a few times the width of the window $g$, $T_v M_{-u} g$ puts
almost all its mass on $\mathbb{R}_+$, and the indicator in
\eqref{e:desiderata} can be dropped. The Hermite basis then satisfies our
requirements, giving the planar GAF of Section~\ref{s:complex}. Intuitively, far
from the real axis, the spectrogram of the analytic white noise will look
like that of proper complex white noise. This heuristic is to relate to standard
time-frequency practice, where one leaves out of the spectrogram a band that is
within the width of the window of the lower half plane. This is meant to avoid
taking into account both positive and negative frequencies of the signal
simultaneously.

%\note[RB]{Describe L function, give closed form if available}
%\note[RB]{Make a quantitive comment on the difference of cdfs for symmetric and
%  planar}

\section{Practical spatial statistics using the zeros of the STFT}
\label{s:stats}
In Section~\ref{s:implementation}, we discuss how to relate the continuous
complex plane $\mathbb{C}$ with the practical discrete implementation of the
Fourier transform. In Section~\ref{s:detection}, we investigate simple
hypothesis tests for signal detection, as in \citep{Fla15}. 

\subsection{Going discrete}
\label{s:implementation}
To fully bridge the gap with numerical signal processing practice, there is an
additional level of approximation that needs to be discussed: Continuous
integrals are replaced by discrete Fourier transforms, so that the fast Fourier
transform can be used. We first describe an experimental setting to study the zeros of the spectrogram of Gaussian white noise. In particular, we explain how to reach an asymptotic regime where the noise occupies an infinite range both in time and frequency and the spectrogram is infinitely well resolved.
% When a signal is present
Second, we investigate practical issues related to detecting a signal in white
noise by using its influence on the distribution of zeros of the spectrogram. 

%\note[RB]{Je crois qu'a ce stade, mieux vaut eviter de partler de discretization steps et de ``specific scales'', car on ne comprend pas encore de quoi il s'agit.}
%the consequences of the presence of signalwhen a signal is present, some specific scales enter in the game. Then one cannot change discretization steps at will anymore.

%%%%
\subsubsection{Zeros of noise only}
\label{s:noiseonly}
Let $F_s$ the sampling frequency, $\Delta t=1/F_s$ the time sampling step size
and $T$ the duration of the observation window. The number of samples is then $N +1$ with 
$N = T/\Delta t$. 

% Approximating $\langle \chi_n,M_v T_u g \rangle$ by $M_vT_ug (n\Delta t)$, it comes
% $$\langle  \fomega, M_{v}T_u g\rangle \approx \lim_{N\rightarrow \infty}
% \sum_{n=1}^N\langle \fomega,\chi_n \rangle e^{-2i\pi v n\Delta t} g(n\Delta t-u),$$
% and we thus expect the discrete spectrogram of a sequence of i.i.d Gaussians
% with large $N$ to be a good approximation to the STFT of white noise.
% }

Let $K$ be
the length of the discretized Gaussian analysis window, i.e. its duration is
$K\Delta t$; therefore $\Delta \nu = F_s/K=1/K\Delta t$ is the frequency
sampling step. In practice, the spectrogram obtained from a discrete STFT is
then an array of size $(N+1, K/2+1)$. Then we consider the time-frequency domain
$[0,T]\times[0,F_s/2]$ only; it corresponds to the analytic signal. This is due
to the Hermitian symmetry of the Fourier transform of real signals: negative
frequencies do not add any information to that carried by positive frequencies,
see also Section~\ref{s:analytic}. This Hermitian symmetry can also be seen on
the zeros of the symmetric GAF in Figure~\ref{f:symmetricPlanarGAF}, where
signal processing practice would have us only consider the upper half-plane ($\nu\geq 0$).
 From \cite{Fel13}'s results, see (\ref{e:countingMeasureSym}), we know that the expected number of zeros of the continuous spectrogram is close to $TF_s/2$ if we neglect the (asymptotically negligible) region $|\nu|\leq a$ close to the time axis, see Section~\ref{s:GAFapproxSYM}. Assuming that we are able to extract every zero, the expected number of zeros in the discrete spectrogram is then $TF_s/2=N/2$ in very good approximation.

Let $\sigma_t=1/(a\sqrt{2\pi})$ and $\sigma_\nu= 1/(2\pi\sigma_t)$
denote the spreads of the Gaussian analysis window $g_a$ in time and frequency, respectively. Note that the scale $a$ serves as a fixed reference for scales in the sequel. We would like to retain the stationary properties of the planar GAF in our discrete
STFTs. We thus require that, in the discrete setting, the resolution --~in number of points~-- should be the same in time and frequency, that is
\begin{equation}
\label{eq_isotropy}
	\frac{\sigma_t}{\Delta t} = \frac{\sigma_\nu}{\Delta \nu} \Longleftrightarrow \sigma_t \cdot F_s = \sigma_\nu \cdot K\Delta t
\end{equation}
This leads to
\begin{equation}
	\left(\frac{\sigma_t}{\Delta t}\right)^2 = \frac{K}{2\pi} \Leftrightarrow \sigma_t = \sqrt{\frac{K}{2\pi}}\Delta t.
\end{equation}
If we want to study the spectrogram of continuous white noise over an infinite time-frequency domain, numerical simulations must obey two necessary conditions:
\begin{equation}
	\begin{cases}
		\text{infinite duration } \Leftrightarrow \text{fine frequency resolution} & : T/\sigma_t = 2\pi\sigma_\nu/\Delta\nu \rightarrow + \infty\\
		\text{infinite frequency range }  \Leftrightarrow \text{fine time resolution} & : F_s/\sigma_\nu=2\pi\sigma_t/\Delta t \rightarrow + \infty\\
		%\text{infinitely fine resolution } & \sigma_t/T \rightarrow 0 \text{ and } \sigma_\nu/F_s \rightarrow 0\\
	\end{cases}
\end{equation}
In terms of samples, these two conditions imply that $N, K \rightarrow \infty$. More precisely, 
\begin{eqnarray}
	\frac{\sigma_t}{T} & = & \frac{1}{N}\sqrt{\frac{K}{2\pi}} \rightarrow 0 \text{ as } N, K \rightarrow \infty\\
	\frac{\sigma_\nu}{F_s} & = & \frac{1}{\sqrt{2\pi K}}\to 0 \text{ as } N, K \rightarrow \infty.
\end{eqnarray}

%% Bilan
These conditions are directly satisfied for $K \propto N$, where $\propto$ means
``proportional to''. Note that in practice because of border effects one chooses $N = 2K$ and keeps the $N$ samples whose time index $n$ is such that $K/2 \leq n \leq N - K/2$. Then, $\sigma_\nu/F_s=1/\sqrt{2\pi K}\propto 1/\sqrt{N}$, $\sigma_t/T\propto 1/\sqrt{N}$; note that $\Delta t/\sigma_t=\Delta \nu/\sigma_\nu\propto 1/\sqrt{N}$ as well. As a result, simulations can asymptotically well approximate the continuous spectrogram of Gaussian white noise.
%% Figure 3 and 4

\begin{figure}
\begin{center}
	\includegraphics[height=60mm]{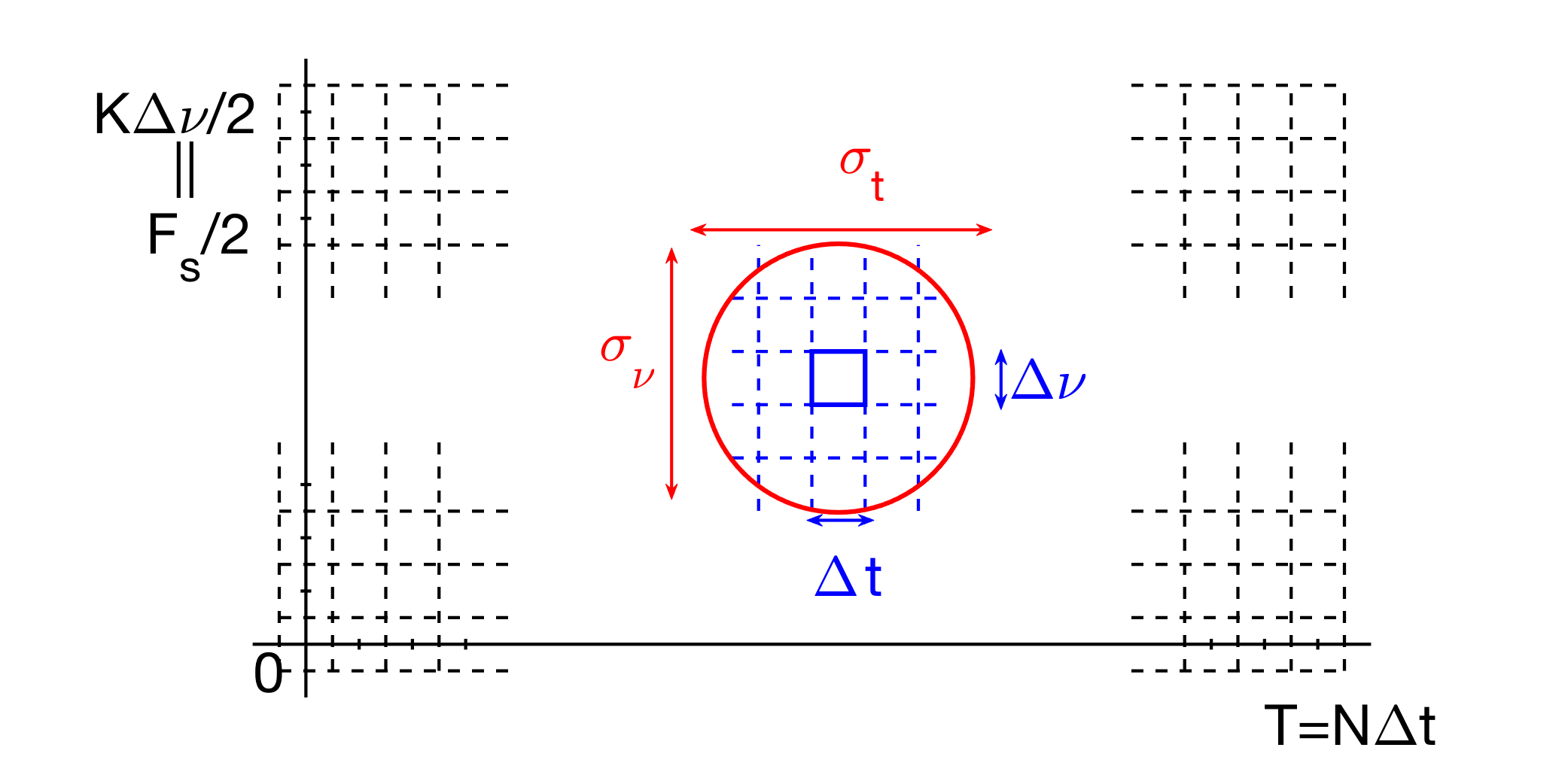}
	\caption{\label{fig_discrete} Illustration of the discrete time-frequency plane $\{(n\Delta t, k\Delta \nu),\; 0\leq n\leq N-1,\; 0\leq k\leq K/2\}$. The resolution of the spectrogram is controlled by the analysis window's Gabor parameters $(\sigma_t, \sigma_\nu)$.}
\end{center}
\end{figure}
\begin{figure}
	\centering
	\includegraphics[height=80mm]{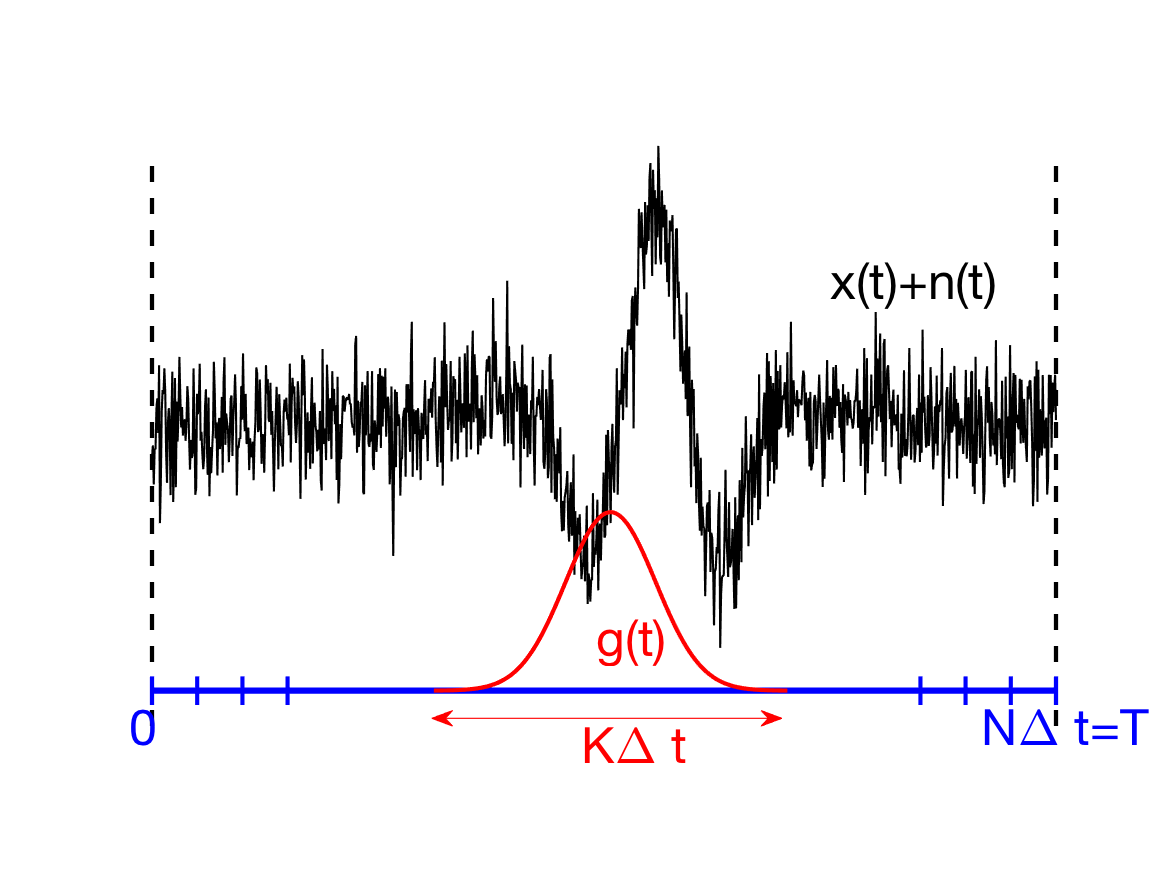}
\caption{\label{fig_stft} Illustration of the STFT: the noisy signal is convolved with a Gaussian
that is translated in time and frequency. The colour code corresponds to
Figure~\ref{fig_discrete} for ease of reference.}
\end{figure}{}

Figure~\ref{fig_discrete} illustrates the relative scales of the duration $T =
N\Delta t$, the frequency range $K/2\Delta t$ (for $\nu\geq 0$), the time and
frequency resolutions $\Delta t$ and $\Delta \nu$, as well as the resolution of
the time-frequency kernel corresponding to the window $g(t)$ with Gabor spread
$(\sigma_t, \sigma_\nu)$. For the sake of completeness and the reader new to
time-frequency, we include in Figure~\ref{fig_discrete} an illustration of the
STFT of a noisy signal.

%% discussion normalization+correspondence discrete.
Now we detail how to relate the discrete coordinates of a discrete spectrogram
with the continuous complex plane. For a given value of $a$, one has
$\sigma_t=1/(a\sqrt{2\pi})$ and thus making the correspondence between samples
and time-frequency units implies setting $\Delta t = \sqrt{2\pi/K}\sigma_t$. For
$a = 1$ one has $\Delta t = \sqrt{1/K}$ so that $u = n/\sqrt{K}$ and $v =
k/\sqrt{K}$ are the coordinates of the time-frequency plane corresponding to
time sample $n$ and frequency sample $k$, respectively.
Figure~\ref{fig:edge_effects} depicts the whole numerical simulation procedure.
It represents the simulated spectrogram and the corresponding extracted area,
taking border effects in consideration. The bound $\ell$ fixes how many samples
close to the zero-frequency axis should be removed. For $a=1$, we have chosen
$\ell = \sqrt{K}$, at it corresponds to $y = 1$ in (\ref{e:countingMeasureSym}).
Note also that border effects alone would actually allow us to extend the shaded
square in Figure~\ref{fig:edge_effects} on its left and right to include $K$
samples. Instead, we chose to reduce it to $K/2-\ell$ mostly for esthetical concerns:
since the point process we observe is almost stationary when only noise is
present, we favoured a square window rather than a rectangle.

% Figure 5
\begin{figure}
	\centering 
	\includegraphics[width=\textwidth]{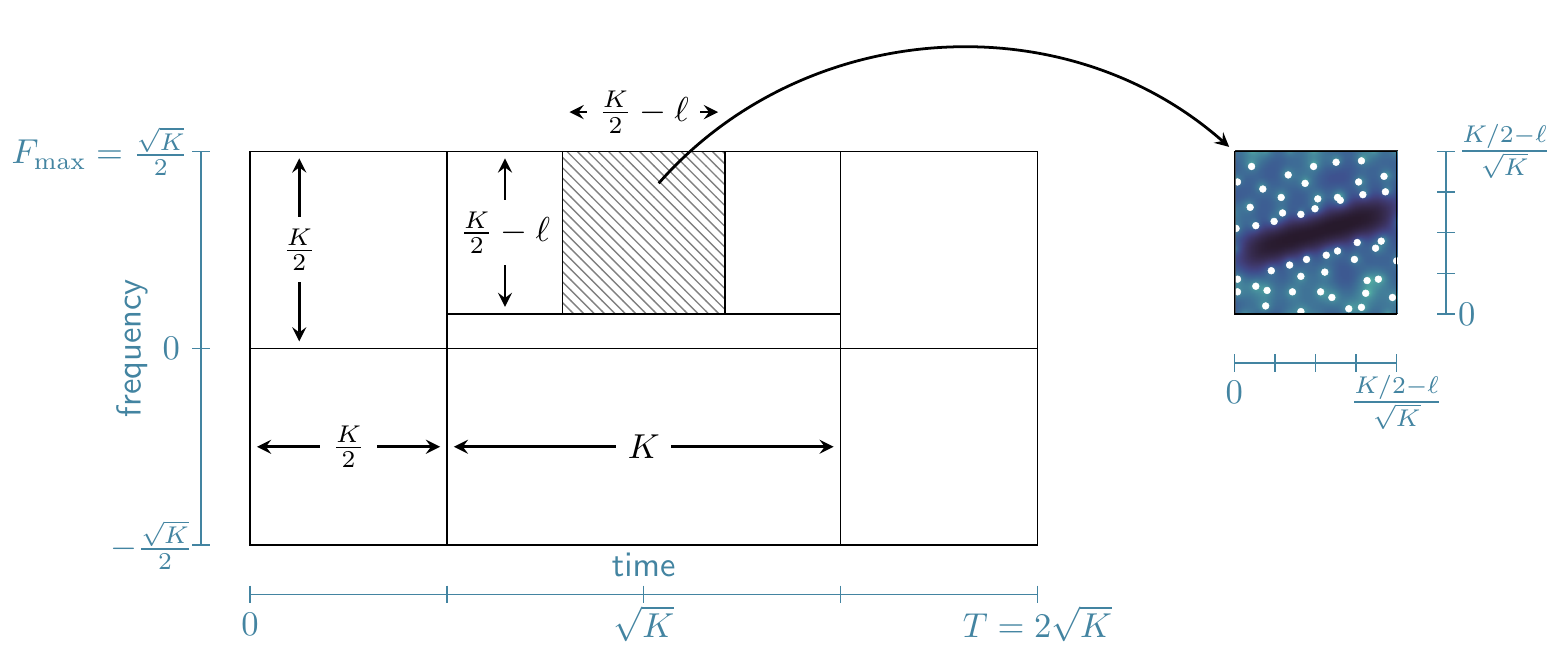}
\caption{Numerical
    simulation procedure. Black ticks indicate the number of samples, while blue
    ticks show time-frequency units for a choice of $\Delta t = 1/\sqrt{K}$ (see
    text for details). In other words, blue ticks are the coordinates in the
    complex plane that are implicit in the mathematical results of
    Sections~\ref{s:real} and \ref{s:complex}. The dashed region corresponds to the area used in subsequent simulations.}\label{fig:edge_effects}
\end{figure}

% discrete spectro
When the conditions above are satisfied, several phenomena occur in the limit of
infinite oversampling $N\to\infty$, which is equivalent to letting both the
duration $T$ and the sampling
frequency $F_s$ grow to infinity. In a dual manner, the resolution $(\Delta t, \Delta \nu)$ of the
discrete spectrogram tends to zero. The time-frequency extent $(\sigma_t, \sigma_\nu)$ of
the analysis window remains constant but is described by a number of samples that grows as $\sigma_t/\Delta t \propto\sqrt{N}$ while $\sigma_t/T\propto 1/\sqrt{N}\to 0$. The analysis window is thus more and more finely resolved, and
we become close to a continuous description.
% zeros
In parallel, the expected number of zeros in the spectrogram of the white noise is $F_sT/2$ and
tends to $\infty$ as $N$ grows. Therefore, assuming perfect zero detection, statistics such as Ripley's $K$ function
or the variance-stabilized $L$ functional statistic of Section~\ref{s:LAndK} can
be asymptotically perfectly well estimated. 

In practice, we defined a numerical
zero as a local minimum among its eight neighbouring bins, and found that the
number of zeros was consistent with what we expected from
Proposition~\ref{p:propertiesOfPlanarGAF}, even if we did not impose a threshold
on the value of the spectrogram at the local minimum. 

We leave this section on a mathematical note. In this section, we implicitly assumed that in the limit on an infinite
observation window and an infinite sampling rate, the discrete Fourier
transforms involved in the computation of the discrete spectrogram converge to
their continuous counterpart. For the sake of completeness, we mathematically
justify in what sense this convergence can be expected. With the notation of Section~\ref{s:real},
subdivide again $[0,T]$ into $N$ equal intervals and denote by $\chi_{n}$ the indicator of the $n$th
interval $[(n-1)\Delta t, n\Delta t]$. Let $P_{N,T}:\cS\rightarrow L^2$ attach
to a Schwartz function $f$ the ``sampled'' simple function $\sum_{n=1}^N
f(n)\chi_n$. Then $P_{N,T}f\rightarrow f$ in $L^2$ as $T$ and $N$ go to infinity
and $T/\sqrt{N}\rightarrow \alpha>0$, which is the setting described above in
this section. On the other hand, 
\begin{equation}
\label{e:dft}
\langle \fomega, P_{N,T}M_vT_u g\rangle = \sum_{n=1}^N\langle \fomega,\chi_n
\rangle e^{-2i\pi v n\Delta t} g(n\Delta t-u)
\end{equation}
is what we call the discrete STFT at $(u,v)$ of a realization of white
noise. Note that in distribution, $(\langle \fomega,\chi_n \rangle)_n$ is a
sequence of i.i.d. Gaussians with variance $\Delta t$. To see how \eqref{e:dft} is a good
approximation to our initial continuous STFT, we note that for all $u,v$,
\begin{eqnarray*}
\mathbb{E}_{\mu_1}\vert\langle \fomega,  M_{v}T_u g\rangle - \langle \fomega,
  P_{N,T} M_{v}T_u g\rangle\vert^2 &=& \mathbb{E}_{\mu_1}\vert\langle \fomega,
                                       M_{v}T_u g
  - P_{N,T} M_{v}T_u g\rangle\vert^2\\
&=& \Vert  M_{v}T_u g
  - P_{N,T} M_{v}T_u g \Vert^2_{L_2} \rightarrow 0.
\end{eqnarray*}

%% When a signal is present  %% ICI
\subsubsection{Zeros of signal plus noise}
\label{s:signalplusnoise}

When a signal is present, its specific scales destroy the scale invariance
property of Gaussian white noise and deprives us from any asymptotic regime in
our numerical simulations. Let $A_S$ denote the typical time and frequency area
occupied by the considered signal. The presence of this signal creates a region
of the spectrogram of size $A_S$ where a decrease in the number of zeros is
expected due to the positive amount of energy corresponding to the signal. This
decrease is clearly visible in the spectrograms of Figure~\ref{f:testPower} for
linear chirps with various $A_S$ and various signal-to-noise ratios (SNR).
% Basis for future tests
The approach proposed here to build statistical detection tests is based on this
intuition. To this purpose one needs to quantify how far the presence of a
signal can influence the statistics used in our tests so that we can maximize
this influence and the efficiency of the proposed test.

% F_s, T...number of zeros and \eta_t, \eta_\nu...
Given a sampling rate $F_s$ and a duration of observation $T$, the unit
intensity in Proposition~\ref{p:propertiesOfPlanarGAF} yields that the expected
number of zeros in the spectrogram of a real white noise is $F_s\cdot T/2 = N/2$, neglecting what happens at small frequencies close to the time axis. Note that this is independent of the width $(\sigma_t,\sigma_\nu)$ of the Gaussian
analysis window $g$. % since it is proportional to $1/\sigma_t\sigma_\nu=2\pi$. 
If one wants to increase the number of zeros in the spectrogram to get
better statistics, it is enough to increase either $F_s$ or $T$. However, the
expected decrease in the number of zeros due to the presence of a signal is of
the order of the area $A_S$, the finite time-frequency area $A_S$ corresponding
to the spectrogram of the signal alone. As a consequence, an excessive increase
in either $F_s$ and/or $T$ would result in an asymptotically complete dilution
of the influence of the signal on the considered statistics. Thus, our purpose
is to build statistics over one or more patches $P$ of the spectrogram of maximal area
$A_P=\eta_t\eta_\nu$ such that $A_S/A_P\simeq 1$. On one hand, a maximal area
$A_P$ is necessary to ensure that the estimate of the chosen statistic be as
accurate as possible (in particular in the presence of noise only, to take into
account as many zeros as possible and minimize the false positive detection
rate); on the other hand, this statistic will be more sensitive to the presence
of a signal if it mostly depends on the influence of the signal on the
distribution of zeros in the spectrogram (in particular, in the presence of signal, we maximize
the true positive detection rate). In practice, note that one can hope to detect only signals
such that $A_S\gg \sigma_t\sigma_\nu=1/2\pi$, which means signals with a
time-frequency support that affects more than $\sigma_t/\Delta t\cdot \sigma_\nu/\Delta \nu = K/2\pi$ samples of the spectrogram.

%The proposed approach cannot be efficient if tests are built on statistics that mostly depend on the noise and will be more efficient if these statistics are significantly 

%characterized by $(\sigma_t,\sigma_\nu = 1/2\pi\sigma_t)$. To ensure that the footprint of $g$ in the spectrogram is isotropic, we have imposed~(\ref{eq_isotropy}) which is equivalent to $F_s/\sigma_\nu = T/\sigma_t$.

%

% \subsection{Sampling GAFs}
% \label{s:sampling}
% The convergence of the truncated GAFs $p_N(z) = \sum_{k=0}^N
% a_kz^k$ in \eqref{e:symmetricPlanarGAF} and \eqref{e:planarGAF} is uniform on
% compact sets. By Hurwitz' theorem \cite{}, the zeros of the truncated GAF are a
% good approximation to those of the GAF \textcolor{red}{I will make this precise later.} 

% Simulating the zeros of truncated GAFs $p_N(z) = \sum_{k=0}^N
% a_kz^k$ can be done through various methods. We have found it more stable to
% diagonalize companion matrices, noting that the roots of the polynomial
% $p_N(z)$ are the eigenvalues of \textcolor{red}{XXX, also in practice, this approach
%   has some limit. FFT on discrete white noise seems to be the most scalable in
%   the end. Not sure this section is necessary anymore, actually.}.

\subsection{Detecting signals through hypothesis testing}
\label{s:detection}

\subsubsection{Monte Carlo envelope tests}
\label{s:alpha}
In Section~\ref{s:LAndK}, we reviewed some popular functional statistics for
stationary isotropic point processes. We focus here on $L$, the
variance-stabilized version of Ripley's $K$ function, and the empty space
function $F$, see Section~\ref{s:spatial}. We follow classical Monte Carlo testing methodology
based on functional statistics, which we now sketch, see e.g. \citep{BDHLMN14}
for a less concise introduction. 

The methodology is independent of the test statistic used, so we introduce it
for a general functional statistic $r\mapsto S(r)$, which we later instantiate
to be $L$ or $F$. Let $\hat S$ denote an
empirical estimate obtained from the spectrogram of data, possibly using edge corrections, see \citep{MoWa03}.
Let $S_0$ be the theoretical functional statistic corresponding to complex white noise. For
$S=L$, $L_0$ can be easily computed from \eqref{e:pairCorrelation}. Note that
our noise is real white noise in the applications, but we approximate the corresponding
2-point correlation function by that of complex white noise far from the real axis, as explained in Section~\ref{s:GAFapproxSYM}. Detection of
signal over white noise can be formulated as testing the hypothesis $H_0$ that
$\hat S$ was built from a realization of a real white noise, versus the
alternate hypothesis $H_1$ that it was not. To do this, we review Monte Carlo
envelope-based hypothesis tests, which are popular across applications. 

In a Monte Carlo envelope test, we define a test statistic $T\in\mathbb{R}$ that
summarizes the difference $r\mapsto S(r)-S_{0}(r)$ in a
single real number, for instance a norm 
\begin{equation}
\label{e:norms}
T_\infty= \sup_{r\in [r_{\min},r_{\max}]}\vert S(r)-S_0(r)\vert\qquad\text{  or
}\qquad T_2=\sqrt{\int_{r_{\min}}^{r_{\max}} \vert S-S_0\vert^2}.
\end{equation} 
Let $t_{\text{exp}}$ denote the realization of $T$ corresponding to the experimental data to be
analyzed. The test consists in simulating $m$ realizations of white noise,
obtaining the corresponding functional statistics estimates $S_1,\dots,S_m$,
computing the realizations $t_1,\dots,t_m$ of the test statistic, and rejecting $H_0$ whenever
the observed $t_{\text{exp}}$ is larger than the $k$-th largest value among
$t_1,\dots,t_m$. Without loss of generality, we assume $t_1,\dots,t_m$ are
in decreasing order, so that $t_k$ is the $k$-th largest. Symmetry
considerations show that this test has significance level $\alpha=k/(m+1)$. When
$S_0$ is not available in closed form, one can replace it by a pointwise average
\begin{equation}
\label{e:average}
\bar{S}_0(r) = \frac{1}{m+1}(S_1(r)+\dots+S_m(r)+\hat{S}(r))
\end{equation}
while preserving the significance level, see \citep{BDHLMN14}.

To see why this test is called an \emph{envelope test}, let $k=10$ and
$m=199$ so that $\alpha=0.05$. We use as a signal a synthetic chirp plus white
noise as in Figure~\ref{f:testPower}, with SNR$=20$. In Figure~\ref{f:envelope}, we take
$r_{\min}=0$ and let $r_{\max}$ vary, showing for each $r_{\max}$ the corresponding $t_{k}$ as the upper limit of the green
shaded envelope. The black line shows $t_{\text{exp}}$ at each $r_{\max}$, for the
same realizations of the tested signal and the white noise spectrograms. To
interpret this plot, imagine the user had fixed $r_{\max}$ to some value, then
he would have rejected $H_0$ if and only if the corresponding intersection of
the black line with $r=r_{\max}$ was above the green area. Note that the
significance of the test in only guaranteed if $r_{\max}$ is fixed prior to
observing data or simulations. Still, Figure~\ref{f:envelope} gives a heuristic
to identify characteristic \emph{scales of interaction} after $H_0$ is rejected.
For instance, characteristic scales could be values of $r_{\max}$ where the data
curve in black leaves the green envelope\footnote{Caveats have been issued against overinterpreting these scales of
  interaction, see \citep{BDHLMN14}.}. The user can thus identify regions of the
spectrogram that possibly correspond to signal (defined as "different from white
noise"). 
To illustrate this, consider again both plots of
Figure~\ref{f:envelope}. There is a hint of an interaction --~an excess or
deficit of pairs-- between
$r_{\max}=0.5$ and $r_{\max}=1$, and this interaction cannot be explained by
noise only. Although we do not delve further here and rather focus on how the
power of the test varies with parameters, this scale can be used to
filter out the noise, in the manner of the Delaunay-based filtering of \cite{Fla15}.
%In
%Figure~\ref{f:}, we plot the spectrogram of the signal\note[RB]{Here whe should
%  have a figure with red circles if it is easy to have one} and highlight the
%pairs of zeros that are at a distance between $0.5$ and $1$. Clearly, the chirp
%is located in the highlighted area. There are various ways to extend this
%procedure by locating the signal and make it a reconstruction algorithm, which
%we leave to future work. Note also that the whole procedure is just a spatial statistics-flavoured take on the Delaunay-based filtering of \cite{Fla15}.

\begin{figure}
\subfigure[$T_\infty$]{
\includegraphics[width=\twofig]{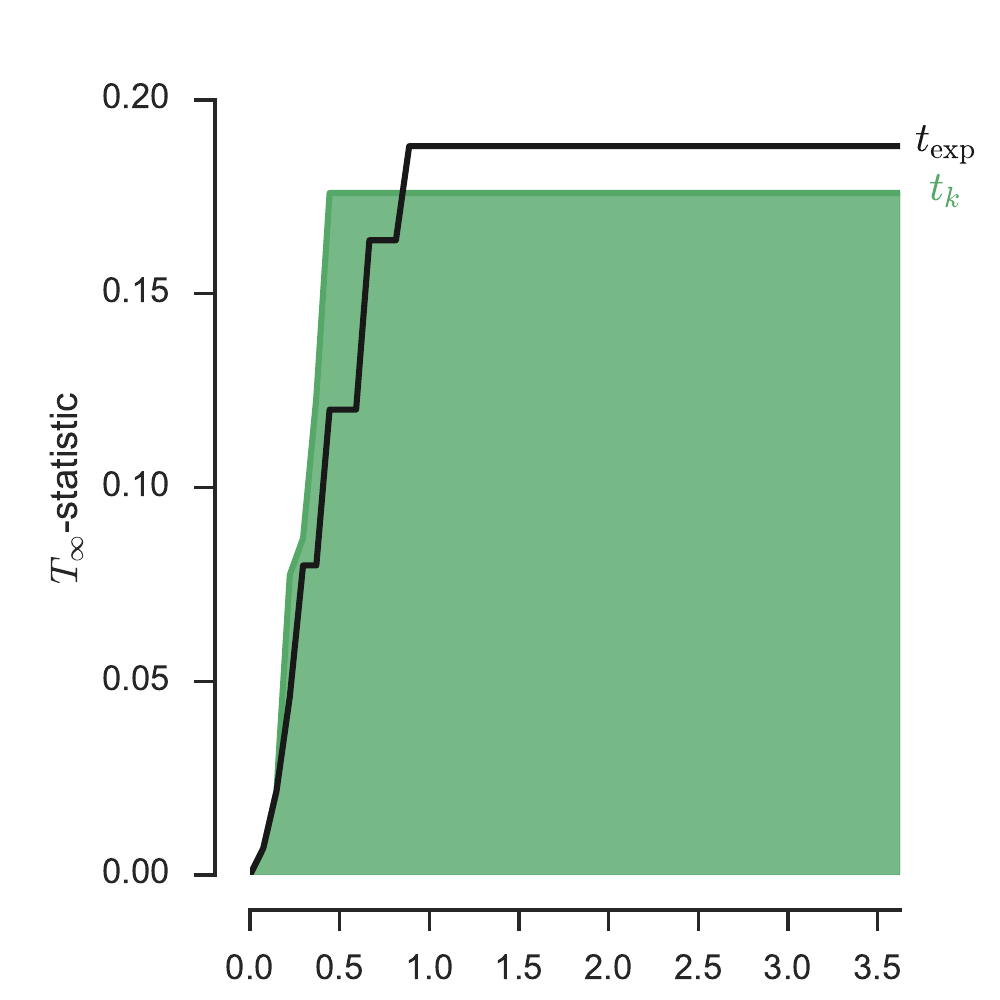}
\label{e:envelpe:tinf}
}
\subfigure[$T_2$]{
\includegraphics[width=\twofig]{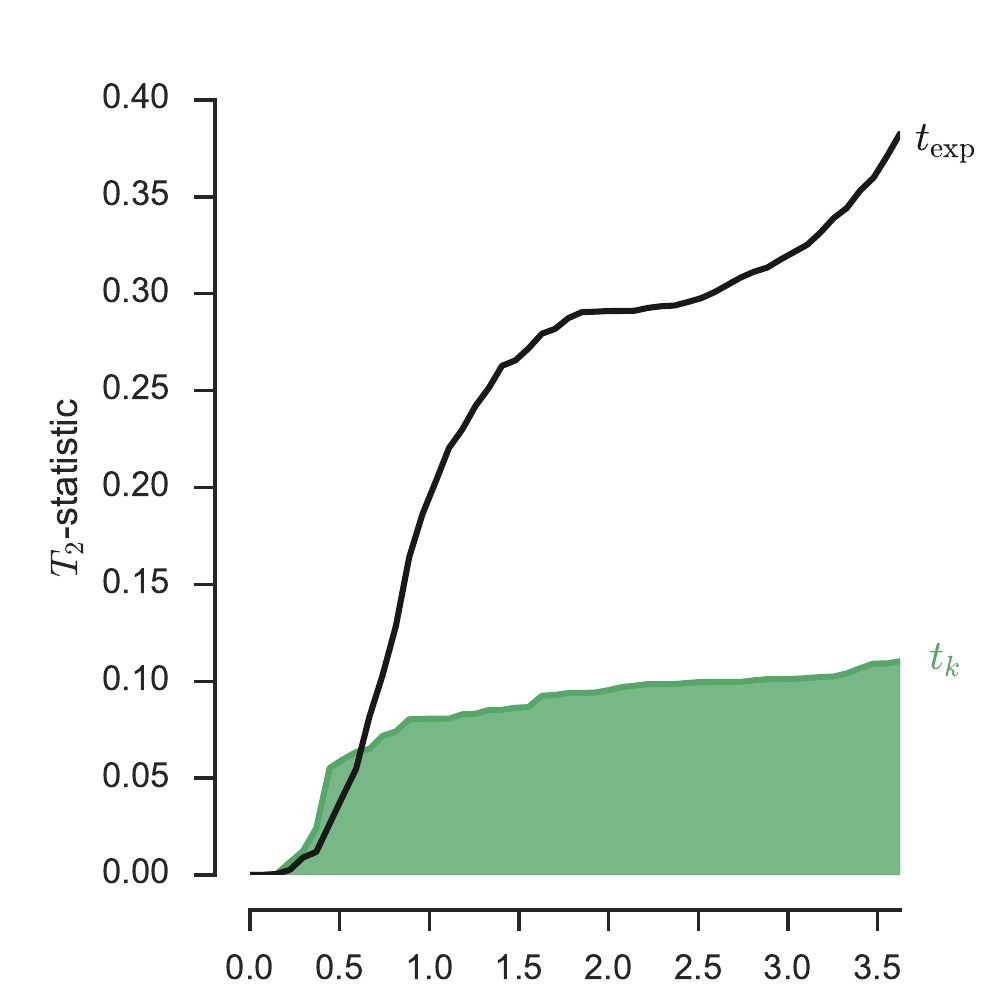}
\label{e:envelpe:t2}
}
\caption{Envelope plots for the detection test of Section~\ref{s:detection} for
  the supremum and 2-norm of the deviation of the $L$ functional statistic from
  its pointwise average \eqref{e:average}.}
\label{f:envelope}
\end{figure}

\subsubsection{Assessing the power of the test}
\label{s:beta}
The significance $\alpha$ of the test --~the probability of rejecting $H_0$
while $H_0$ is true~-- is fixed by the user as in Section~\ref{s:alpha}. It
remains to investigate the power $\beta$ of the test, that is, the probability
of rejecting $H_0$ when one should. Following Section~\ref{s:signalplusnoise},
we expect $\beta$ to increase with SNR, which should be large enough to ``push''
zeros away from the time-frequency support of the signal to be detected. We also
expect the power to be larger when the observation window is not too much larger
than the time-frequency support $A_S$ of the signal.

% SNR, so that to decrease when the 
% We expect $\beta$ to depend on the bias and variance of $\hat{S}$ under $H_0$:
% the green area in Figure~\ref{f:envelope} should be of small width. This width
% diminishes as the size $N$ of the observation window in Figure~\ref{f:} grows.
% Yet at the same time, if we are to detect deficits or excesses of pairs of
% points using cumulative statistics like the $L$ or $F$ function in
% Section~\ref{s:spatial}, then the number of pairs that contribute to departing
% from $H_0$ should be large enough when compared to the total number of pairs.
% In practice, SNR should thus be large enough to ``push'' zeros away
% from the time-frequency support of the signal to be detected. Also, the size of
% the observation window should not be too large compared to the characteristic
% size of this support.

We back these claims by the experiment in Figure~\ref{f:testPower}, where we
assume signals take the form of linear chirps. Still taking $m=199$ and $k=10$,
so that $\alpha=0.05$, we build each of the six panels as follows: we simulate a
mock signal made using a linear chirp plus noise, with SNR indicated on the
plot, growing from left to right. We then repeat $200$ times: 1) simulate $m$
white noise spectrograms, 2) check wether $H_0$ is rejected for each value
of $r_{\max}$. We can thus estimate the probability $\beta$ of rejecting $H_0$
for various choices of $r_{\max}$ the user could have made. We plot both the
power using $S=L$ or $S=F$, choosing the $2-$norm in \eqref{e:norms} and the
empirical average \eqref{e:average}. We estimate the functional statistics using the \texttt{spatstat} R
package\footnote{Version 1.51-0, see \url{http://spatstat.org/}}. To identify
the statistical significance of our estimated powers, we plot Clopper-Pearson
confidence intervals for $5$ values of $r_{\max}$, using a Bonferroni correction for the $10$ multiple tests involved
on each plot, see e.g. \cite{Was13}. Finally, the top row of
Figure~\ref{f:testPower} corresponds to a signal support that matches the size of
the observation window, while the bottom row is half that. On each panel, an inlaid plot
depicts the spectrogram for one realization of the signal corrupted by white
noise. Spectrogram zeros are in white.

\begin{figure}
\includegraphics[width=\threefig]{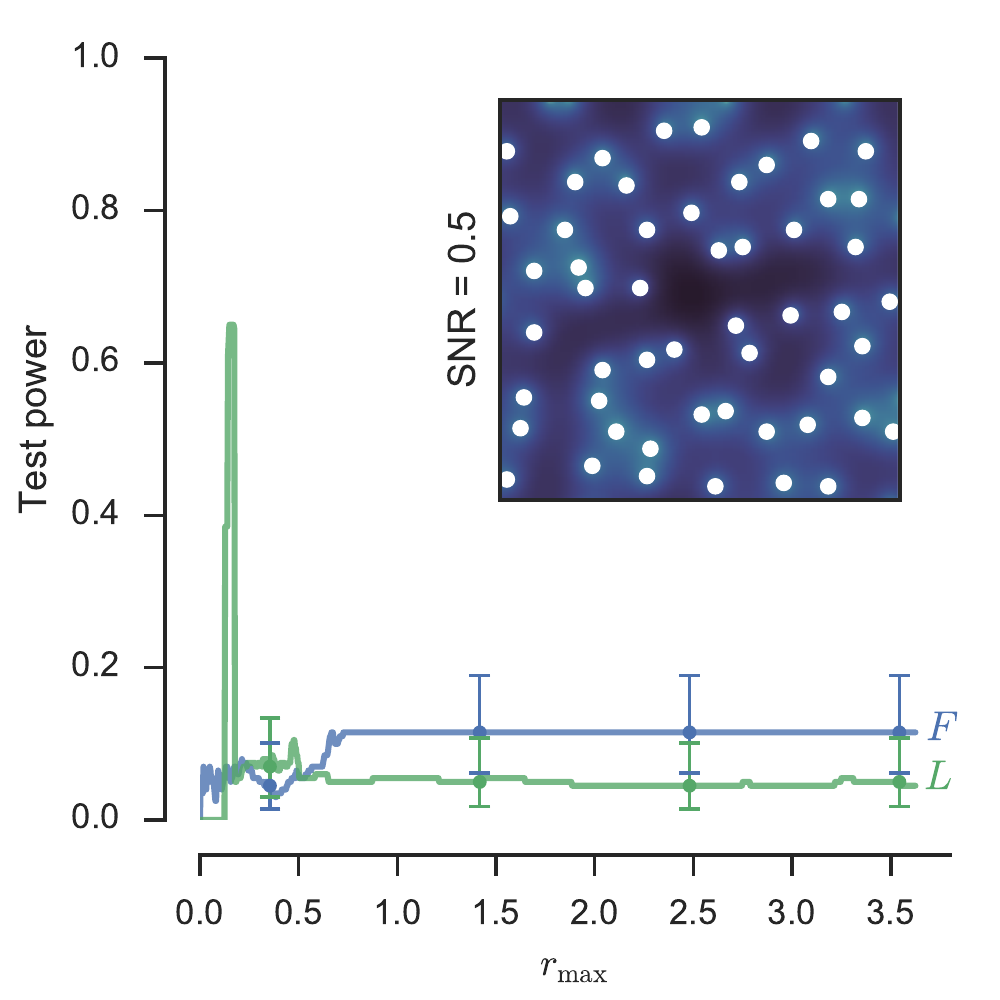}
\includegraphics[width=\threefig]{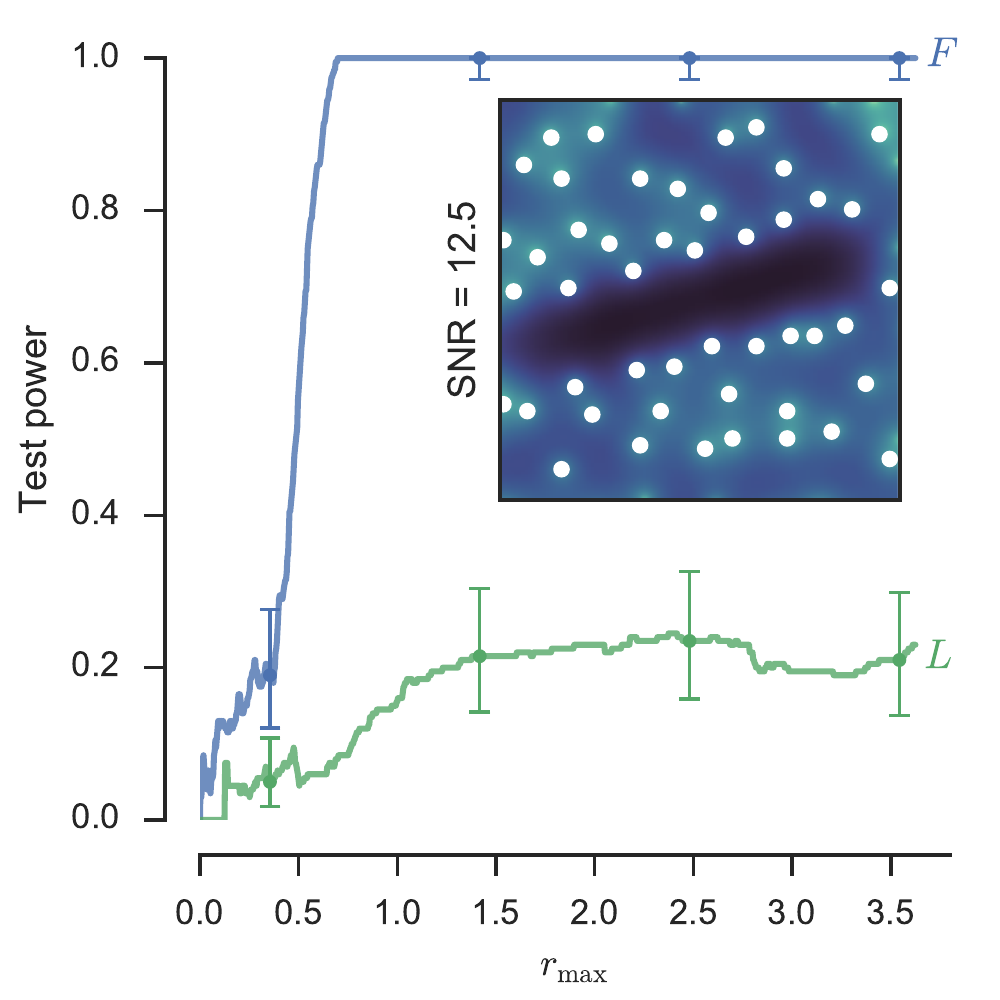}
\includegraphics[width=\threefig]{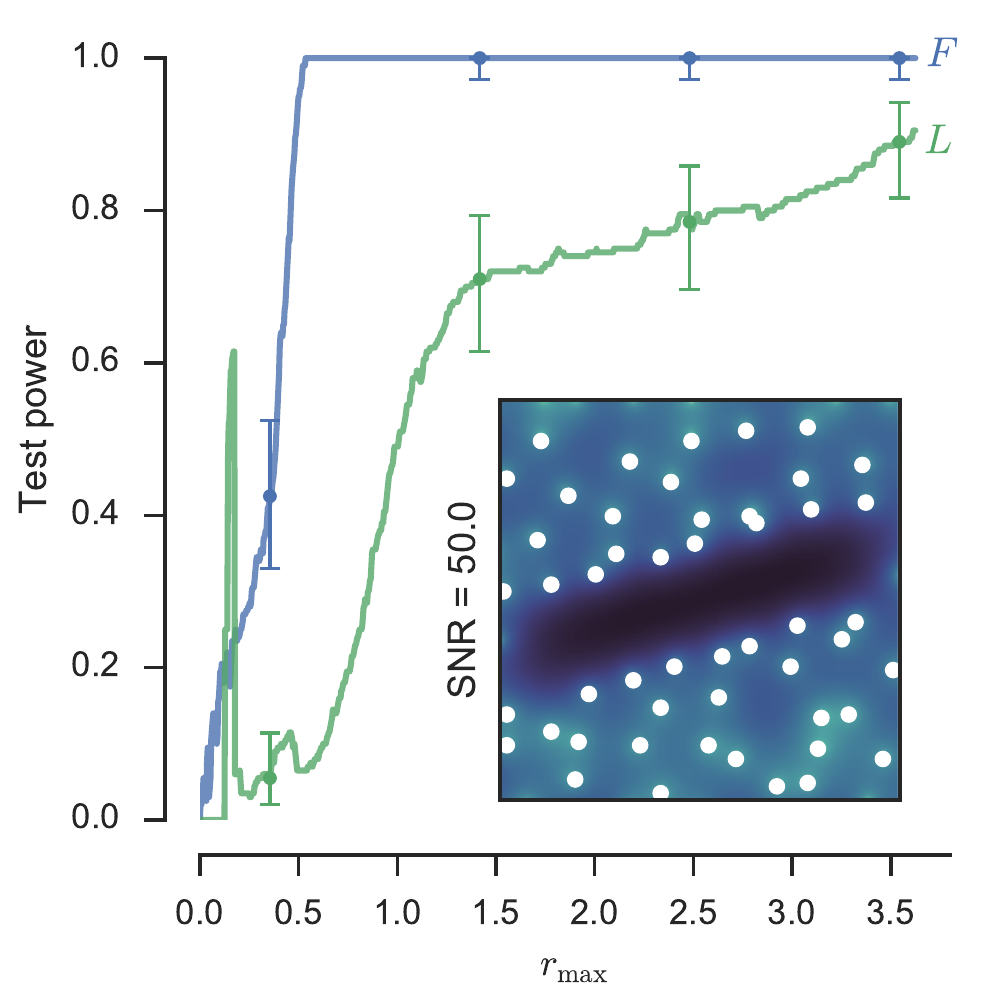}\\
\includegraphics[width=\threefig]{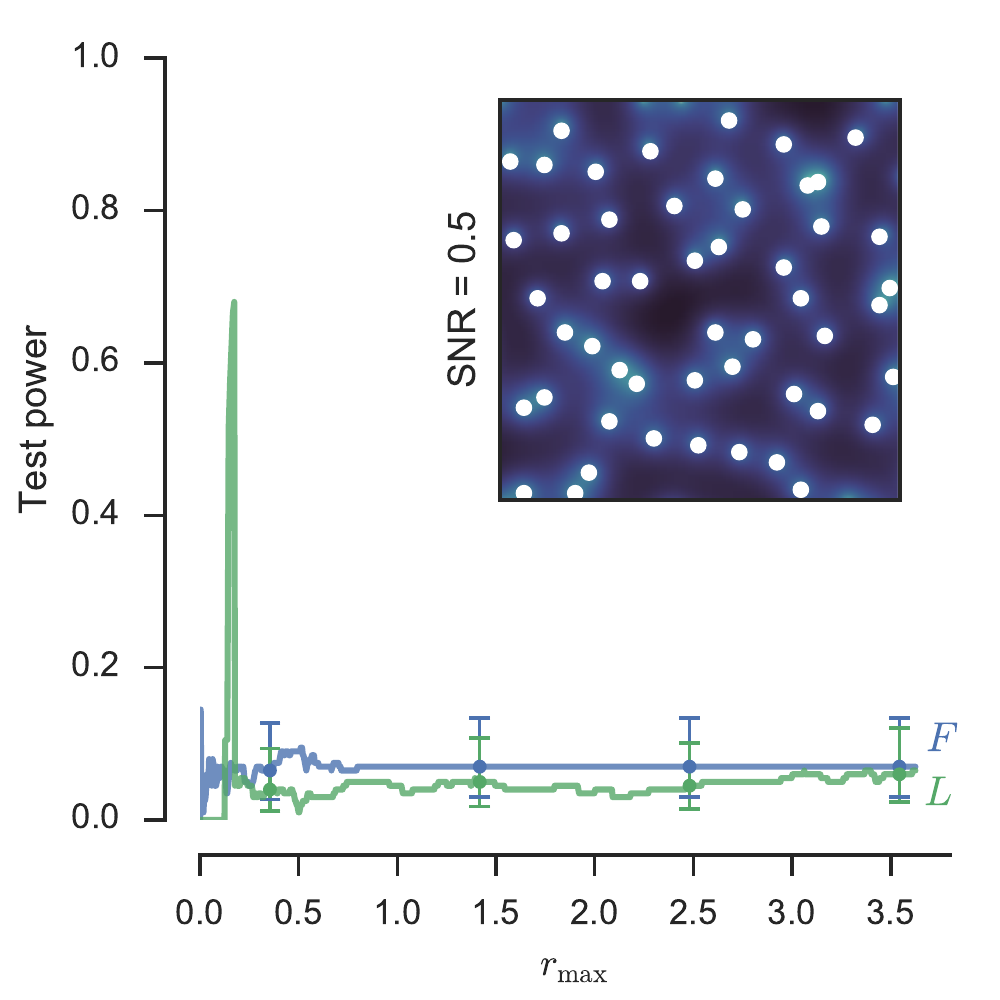}
\includegraphics[width=\threefig]{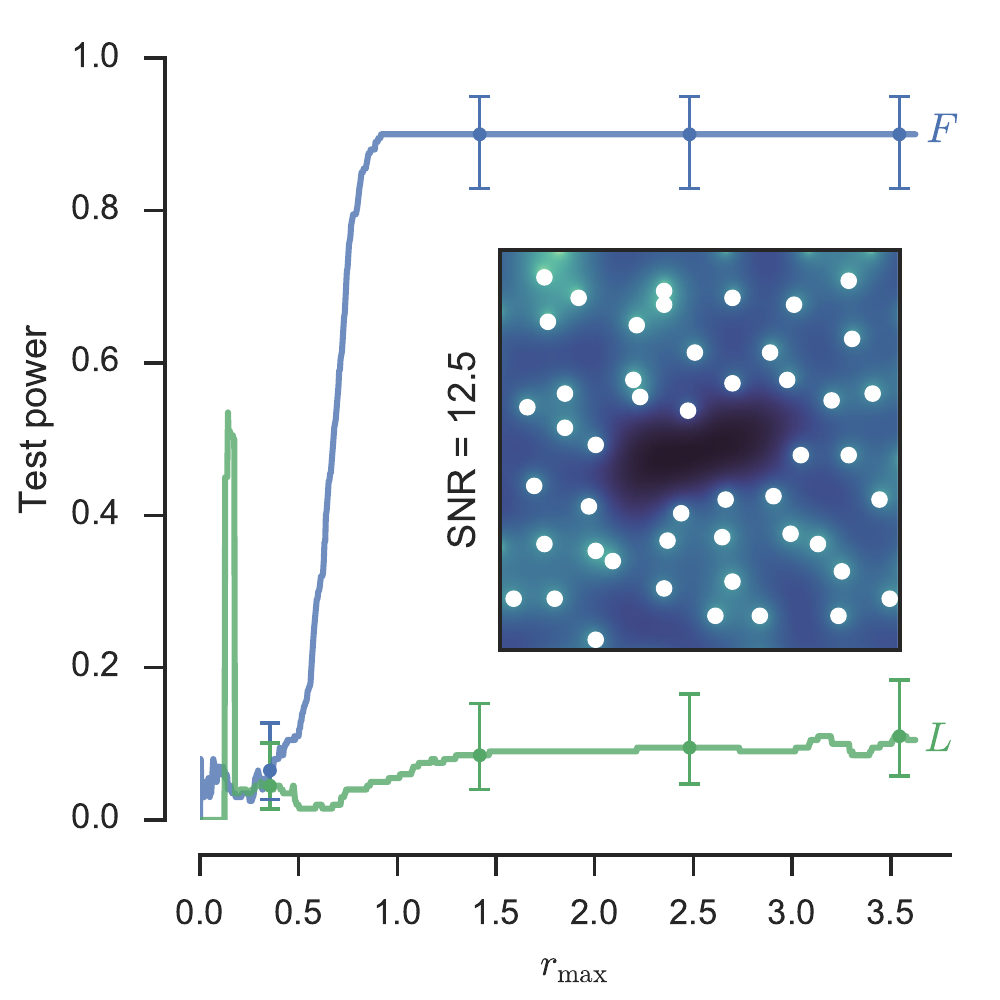}
\includegraphics[width=\threefig]{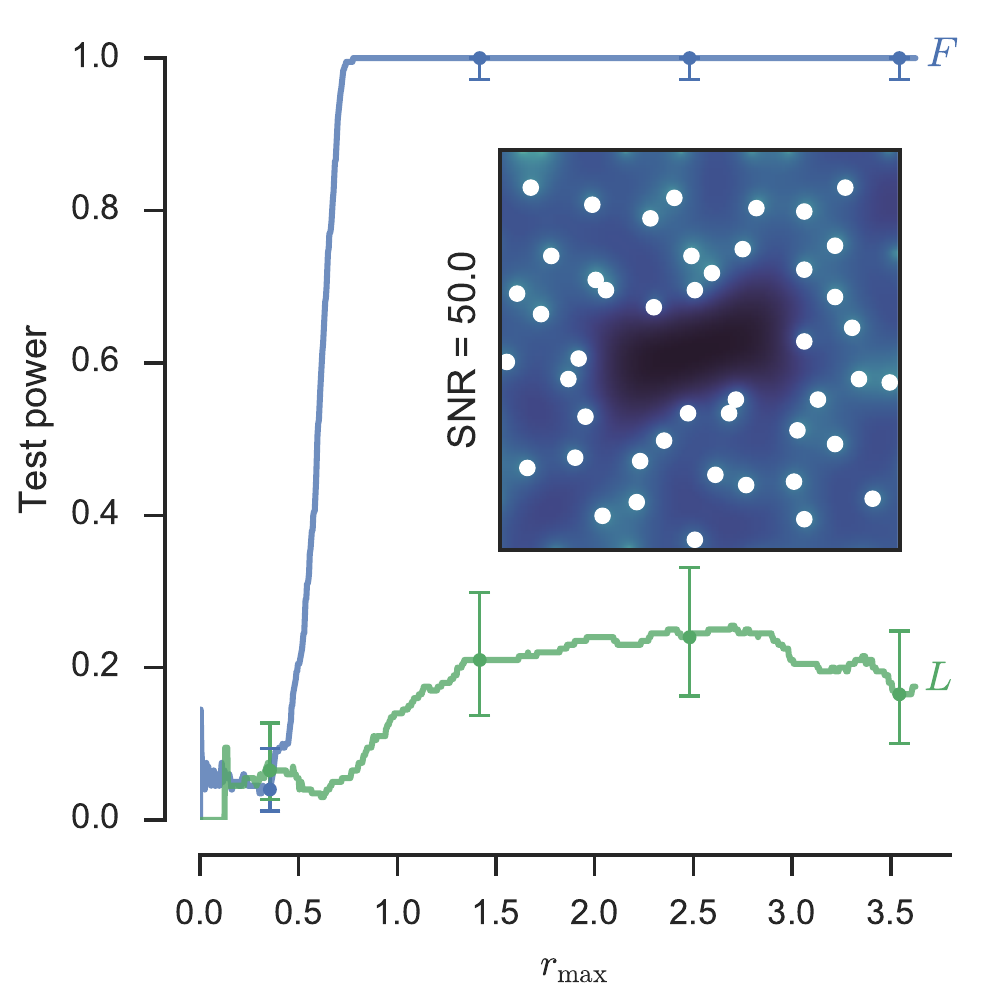}

\caption{
Assessing the power of the test on detecting a linear chirp with various SNRs
across columns. The top row corresponds to a larger support of the chirp
compared to the observation window.
\label{f:testPower}
}
\end{figure} 

The results confirm our intuitions: power increases with SNR, and decreases as
the size of the support of the signal diminishes with respect to the observation
window. In all experiments, the best power is obtained by taking $r_{\max}$ to
be as large as possible, which here means half of the observation window. This
makes sure that as many points/pairs as possible enter the estimation of the functional
statistic $S$. Concerning the choice of functional statistic, the empty space
function $F$ performs significantly better for high SNR and large enough
$r_{\max}$. The green peaks of power at low $r_{\max}$ for some combinations of
SNR and support are due to the excess of small pairwise distances introduced by
the chirp signal. The power vanishes quickly once larger pairwise distances are
considered, due to the cumulative nature of $L$. It is hard to rely on these
peaks as they do not appear systematically and would require a careful hand-tuning of
$r_{\max}$ that would likely defeat our purpose of automatizing detection. So overall, we would recommend using
$F$ and large $r_{\max}$, which appears to be a robust best choice. We also
found (not shown) first that $F$ is superior or equal to the other
functional statistics described in Section~\ref{s:spatial} for chirp detection. Second, we found that the tests using the average \eqref{e:average}
are consistently more powerful than those using the analytic form $L_0$ of
$L$. We believe this is due to the edge correction that is implicitly made in
\eqref{e:average}, while the analytic $L_0$ corresponds to an infinite
observation window. Third, we also observed the $2-$norm in \eqref{e:norms} to
be consistently more powerful than the supremum norm.

\section{Discussion}
\label{s:discussion}
We showed how to give a mathematical meaning to the zeros of the spectrogram of
white noise, and investigated their statistical distribution for real, complex,
and --~to a lesser extent~-- analytical white noise. We have related these zeros
to the zeros of Gaussian analytic functions, a topic of booming interest in
probability. More pragmatically, we investigated the computational issues raised
by implementing tests based on spectrogram zeros. 

The connection with GAFs puts signal processing algorithms based on
spectrogram zeros on firm ground, and further progress on GAFs is bound to be
fruitful for signal processing. Perhaps less obviously, we believe
signal processing tools can also bring insight into probabilistic questions on
GAFs. For starters, the Bargmann transform, spectrogram zeros and the fast Fourier transform give a novel way to approximately simulate the zeros of the planar GAF, or even the zeros of random polynomials.

As for the detection of signals using spectrogram
zeros, we have investigated the application of standard frequentist testing
tools. They showed good power for high SNR, but the performance decreases for
low SNR and small signal support compared to the observation window. 
There are various leads to improve on these two points. First, we could
transform our global test into several local tests, trying to adapt the tested
patch to the support of the signal. Second, models for signals could be fed to
Bayesian techniques, allowing to explore all signals compatible with a
given pattern of zeros.

\subsection*{Acknowledgments}
We thank Patrick Flandrin, Adrien Hardy, and Fred Lavancier for fruitful
discussions on various aspects of this paper. RB acknowledges support from ANR
\textsc{BoB} (ANR-16-CE23-0003), and all authors acknowledge support from
ANR \textsc{BNPSI} (ANR-13-BS03-0006).
\bibliography{stft,stats}

\begin{thebibliography}{24}
\providecommand{\natexlab}[1]{#1}
\providecommand{\url}[1]{\texttt{#1}}
\expandafter\ifx\csname urlstyle\endcsname\relax
  \providecommand{\doi}[1]{doi: #1}\else
  \providecommand{\doi}{doi: \begingroup \urlstyle{rm}\Url}\fi

\bibitem[Baddeley et~al.(2014)Baddeley, Diggle, Hardegen, Lawrence, Milne, and
  Nair]{BDHLMN14}
A.~Baddeley, P.~J. Diggle, A.~Hardegen, T.~Lawrence, R.~K. Milne, and G.~Nair.
\newblock On tests of spatial pattern based on simulation envelopes.
\newblock \emph{Ecological Monographs}, 84\penalty0 (3):\penalty0 477--489,
  2014.

\bibitem[Cohen(1995)]{Coh95}
L.~Cohen.
\newblock \emph{Time-frequency analysis}, volume 778.
\newblock Prentice Hall PTR Englewood Cliffs, NJ:, 1995.

\bibitem[Daley and Vere-Jones(2003)]{DaVe03}
D.~J. Daley and D.~Vere-Jones.
\newblock \emph{An Introduction to the Theory of Point Processes}.
\newblock Springer, 2nd edition, 2003.

\bibitem[Feldheim(2013)]{Fel13}
N.~D. Feldheim.
\newblock Zeroes of {G}aussian analytic functions with translation-invariant
  distribution.
\newblock \emph{Israel Journal of Mathematics}, 195\penalty0 (1):\penalty0
  317--345, 2013.

\bibitem[Flandrin(1998)]{Fla98}
P.~Flandrin.
\newblock \emph{Time-frequency/time-scale analysis}, volume~10.
\newblock Academic press, 1998.

\bibitem[Flandrin(2015)]{Fla15}
P.~Flandrin.
\newblock Time--frequency filtering based on spectrogram zeros.
\newblock \emph{IEEE Signal Processing Letters}, 22\penalty0 (11):\penalty0
  2137--2141, 2015.

\bibitem[Flandrin(2017)]{Fla17}
P.~Flandrin.
\newblock On spectrogram local maxima.
\newblock In \emph{International Conference on Acoustics, Speech and Signal
  Processing (ICASSP)}, pages 3979--3983. IEEE, 2017.

\bibitem[Gautschi(2004)]{Gau04}
W.~Gautschi.
\newblock \emph{Orthogonal polynomials: computation and approximation}.
\newblock Oxford University Press, USA, 2004.

\bibitem[Gr{\"o}chenig(2001)]{Gro01}
K.~Gr{\"o}chenig.
\newblock \emph{Foundations of time-frequency analysis}.
\newblock Birkh\"{a}user, 2001.

\bibitem[Hannay(1998)]{Han98}
J.~H. Hannay.
\newblock The chaotic analytic function.
\newblock \emph{Journal of Physics A: Mathematical and General}, 31\penalty0
  (49):\penalty0 L755, 1998.

\bibitem[Holden et~al.(2010)Holden, {\O}ksendal, Ub{\o}e, and Zhang]{HOUZ10}
H.~Holden, B.~{\O}ksendal, J.~Ub{\o}e, and T.~Zhang.
\newblock \emph{Stochastic partial differential equations}.
\newblock Springer, second edition, 2010.

\bibitem[Hough et~al.(2006)Hough, Krishnapur, Peres, and Vir\'ag]{HKPV06}
J.~B. Hough, M.~Krishnapur, Y.~Peres, and B.~Vir\'ag.
\newblock Determinantal processes and independence.
\newblock \emph{Probability surveys}, 2006.

\bibitem[Hough et~al.(2009)Hough, Krishnapur, Peres, and Vir{\'a}g]{HKPV09}
J.~B. Hough, M.~Krishnapur, Y.~Peres, and B.~Vir{\'a}g.
\newblock \emph{Zeros of {G}aussian analytic functions and determinantal point
  processes}, volume~51.
\newblock American Mathematical Society Providence, RI, 2009.

\bibitem[Krishnapur and Vir\'{a}g(2014)]{KrVi14}
M.~Krishnapur and B.~Vir\'{a}g.
\newblock The {G}inibre ensemble and {G}aussian analytic functions.
\newblock \emph{International Mathematics Research Notices}, 2014\penalty0
  (6):\penalty0 1441--1464, 2014.

\bibitem[Lavancier et~al.(2014)Lavancier, M{\o}ller, and Rubak]{LaMoRu14}
F.~Lavancier, J.~M{\o}ller, and E.~Rubak.
\newblock Determinantal point process models and statistical inference.
\newblock \emph{Journal of the Royal Statistical Society}, 2014.

\bibitem[Macchi(1975)]{Mac75}
O.~Macchi.
\newblock The coincidence approach to stochastic point processes.
\newblock \emph{Advances in Applied Probability}, 7:\penalty0 83--122, 1975.

\bibitem[M{\o}ller and Waagepetersen(2003)]{MoWa03}
J.~M{\o}ller and R.~P. Waagepetersen.
\newblock \emph{Statistical inference and simulation for spatial point
  processes}.
\newblock CRC Press, 2003.

\bibitem[Nishry(2010)]{Nis10}
A.~Nishry.
\newblock Asymptotics of the hole probability for zeros of random entire
  functions.
\newblock \emph{International Mathematics Research Notices}, 2010.

\bibitem[Picinbono(1997)]{Pic97}
B.~Picinbono.
\newblock On instantaneous amplitude and phase of signals.
\newblock \emph{IEEE Transactions on Signal Processing}, 1997.

\bibitem[Picinbono and Bondon(1997)]{PiBo97}
B.~Picinbono and P.~Bondon.
\newblock Second-order statistics of complex signals.
\newblock \emph{IEEE Transactions on Signal Processing}, 1997.

\bibitem[Prosen(1996)]{Pro96}
T.~Prosen.
\newblock Exact statistics of complex zeros for {G}aussian random polynomials
  with real coefficients.
\newblock \emph{Journal of Physics A: Mathematical and General}, 29\penalty0
  (15):\penalty0 4417, 1996.

\bibitem[Pugh(1982)]{Pug82}
E.~L. Pugh.
\newblock The generalized analytic signal.
\newblock \emph{Journal of Mathematical Analysis and Applications}, 89\penalty0
  (2):\penalty0 674--699, 1982.

\bibitem[Schehr and Majumdar(2008)]{ScMa08}
G.~Schehr and S.~N. Majumdar.
\newblock Real roots of random polynomials and zero crossing properties of
  diffusion equation.
\newblock \emph{Journal of Statistical Physics}, 132\penalty0 (2):\penalty0
  235--273, 2008.

\bibitem[Wasserman(2013)]{Was13}
L.~Wasserman.
\newblock \emph{All of statistics: a concise course in statistical inference}.
\newblock Springer Science \& Business Media, 2013.

\end{thebibliography}
\bibliographystyle{plainnat}

\end{document}